\documentclass[letterpaper,journal, final, twocolumn]{IEEEtran}

\usepackage{cite}

\usepackage{graphicx}
\usepackage{stmaryrd}

\usepackage[cmex10]{amsmath}
\usepackage{amssymb}
\usepackage{amsthm}

\usepackage{fixltx2e}

\usepackage{booktabs}

\usepackage[caption=false]{subfig}

\usepackage{tikz,pgf}

\usepackage{pgfplots}
\usetikzlibrary{shapes,positioning,arrows,decorations.markings,fit,calc,patterns}
\usepackage{bm}

\usepackage{ifpdf}
\ifpdf
\usepackage[final,pdfborder={0 0 0}]{hyperref}
\fi

\usepackage{color}
\newcommand{\fixme}[2]{\ifx&#2&{\color{red}#1}\else{\color{red}FIXME\{}#1{\color{red}\}}\footnote{{\color{red}#2}}\PackageWarning{Fixme}{#1: #2}\fi}

\newcommand{\Arikan}{Ar\i{}kan}

\newcommand{\binary}[1]{\langle #1 \rangle_2}
\newcommand{\dominates}{\succeq}

\newcommand{\upgraded}{\succeq}
\newcommand{\GF}{\mathrm{GF}}

\newcommand{\Encode}{\mathsf{Encode}}

\newcommand{\be}[1]{\begin{equation}\label{#1}}
  \newcommand{\ee}{\end{equation}}

\DeclareMathAlphabet{\mathbfsl}{OT1}{ppl}{b}{it} 

\theoremstyle{plain} 
 
\newtheorem{thm}{Theorem\hspace{-1pt}} 
\newenvironment{theorem}
{\begin{thm}\hspace*{-1ex}}{\end{thm}}
 
\newtheorem{lem}[thm]{Lemma\hspace{-1.75pt}}
\newenvironment{lemma}{\begin{lem}\hspace*{-1ex}}{\end{lem}}
 
\newtheorem{prop}[thm]{Proposition$\!$}

\newtheorem{cor}[thm]{Corollary$\!$}
\newenvironment{corollary}{\begin{cor}\hspace*{-1ex}}{\end{cor}}
 
\newtheorem{defn}{Definition$\!$}

\newcounter{enumrom}
\renewcommand{\theenumrom}{(\roman{enumrom})}

\newcommand{\bfu}{\mathbf{u}}
\newcommand{\bfv}{\mathbf{v}}
\newcommand{\bfx}{\mathbf{x}}

\newcommand{\reverse}[1]{\overset{{}_{\shortleftarrow}}{#1}}
\newcommand{\GReversed}{G_{\mathrm{rv}}}
\newcommand{\GNonReversed}{G_{\mathrm{nrv}}}
\newcommand{\myset}[1]{\left\{ #1 \right\}}
\newcommand{\mysett}[1]{\{ #1 \}}

\newcommand{\myspan}{\mathrm{span}}
\newcommand{\encoder}{\mathcal{E}}
\newcommand{\encoderReversed}{\encoder_{\mathrm{rv}}}
\newcommand{\encoderNonReversed}{\encoder_{\mathrm{nrv}}}
\newcommand{\then}{\Longrightarrow}
\newcommand{\myfloorr}[1]{\lfloor#1\rfloor}
\newcommand{\calX}{\mathcal{X}}
\newcommand{\calY}{\mathcal{Y}}
\newcommand{\calZ}{\mathcal{Z}}

\begin{document}

\title{Flexible and Low-Complexity Encoding and Decoding of Systematic Polar Codes}
\author{%
    Gabi~Sarkis, %
    Ido~Tal,~\IEEEmembership{Member,~IEEE}, %
    Pascal~Giard,~\IEEEmembership{Student~Member,~IEEE}, %
    Alexander~Vardy,~\IEEEmembership{Fellow,~IEEE}, %
    Claude~Thibeault,~\IEEEmembership{Senior~Member,~IEEE}, and %
    Warren~J.~Gross,~\IEEEmembership{Senior~Member,~IEEE}
\thanks{G. Sarkis, P. Giard, and W. J. Gross are with the Department of Electrical and Computer Engineering, McGill University, Montr\'{e}al, QC H3A 0E9, Canada (e-mails: \{gabi.sarkis, pascal.giard\}@mail.mcgill.ca, warren.gross@mcgill.ca).}
\thanks{I. Tal is with the Technion---Israel Institute of Technology, Haifa 32000, Israel. (e-mail: idotal@ee.technion.ac.il).}
\thanks{A. Vardy is with the Department of Electrical and Computer Engineering
and the Department of Computer Science and Engineering, University of California at San Diego, La Jolla, CA 92093, USA (e-mail: avardy@ucsd.edu).}
\thanks{C. Thibeault is with the Department of Electrical Engineering, \'Ecole de technologie sup\'erieure, Montr\'eal, QC H3C 1K3, Canada (e-mail: claude.thibeault@etsmtl.ca).}}
\maketitle

\begin{abstract}
  In this work, we present hardware and software implementations of flexible polar systematic encoders and decoders. The proposed implementations operate on polar codes of any length less than a maximum and of any rate.
We describe the low-complexity, highly parallel, and flexible systematic-encoding algorithm that we use and prove its correctness.
  Our hardware implementation results show that the overhead of adding code rate and length flexibility is little, and the impact on operation latency minor compared to code-specific versions. Finally, the flexible software encoder and decoder implementations are also shown to be able to maintain high throughput and low latency.
\end{abstract}

\section{Introduction}
\label{sec:intro}
Modern communication systems must cope with varying channel conditions and differing throughput and transmission latency constraints.
The 802.11-2012 wireless communication standard, for example, requires more than twelve error-correction configurations, increasing implementation complexity \cite{Lee2008,Condo2013}. Such a requirement necessitates encoder and decoder implementations that are flexible in code rate and length.

Polar codes achieve the symmetric capacity of memoryless channels with an explicit construction and are decoded with the low-complexity successive-cancellation decoding algorithm \cite{Arikan2009}. In this paper, we show that apart from the above favorable properties, polar codes are highly amenable to flexible encoding and decoding. That is, their regular structure enables encoder and decoder implementations that support any polar code of any rate and length, under the constraint of a maximal codeword length.

Systematic polar coding was described in \cite{Arikan2011} as a method to ease information extraction and improve bit-error rate without affecting the frame-error rate.
The systematic encoding scheme originally proposed in \cite{Arikan2011} is serial by nature, and seems non-trivial to parallelize, unless restricted to a single polar code of fixed rate
and length. The serial nature of this encoding ($O(n \cdot \log n)$ time-complexity, where $n$ is the code length) places a speed limit on the encoding process which gets worse with increasing code length.
To address this, a new systematic encoding algorithm that is easy to parallelize was first described in \cite{Sarkis2014}. This algorithm is both parallel and flexible in code rate. In this work, we extend the flexibility to code length as well and provide hardware and software implementations that achieve throughput values of 29 Gbps and 10 Gbps, respectively.

We dedicate a portion of this work to proving the correctness of the systematic encoding algorithm presented in \cite{Sarkis2014}. We prove that it results in valid systematic polar codewords when the sub-matrix of the encoding matrix with rows and columns corresponding to information bit indices is an involution. We prove that this condition is satisfied for both polar and Reed-Muller codes since they both satisfy a property we call \emph{domination contiguity}, which we prove is a sufficient condition for the involution to be true.

This paper is organized into two parts addressing flexible encoding and decoding, respectively. The first part starts with Section~\ref{sec:background} where we define some preliminary notation and contrast the implementation of the original systematic encoder presented in \cite{Arikan2011} with that of \cite{Sarkis2014}. Note that reading \cite{Sarkis2014} or \cite{Arikan2011} \emph{is not} a prerequisite to reading the current paper, since we summarize the key points needed from those papers in Section~\ref{sec:background}. Section~\ref{sec:systematicEncodingDefs} is mainly about setting notation and casting the various operations needed in matrix form. In Section~\ref{sec:dominationContiguiryImpliesInvolution}, we define the property of domination contiguity, and prove that our algorithm works---in both natural and bit-reversed modes---if this property is satisfied. The fact that domination contiguity indeed holds for polar codes is proved in Section~\ref{sec:polarCodesSatisfyDominationContiguity}. With correctness of the algorithm proved, flexible hardware and software systematic encoder implementations are presented in Sections \ref{sec:flex-enc} and \ref{sec:flex-sw-enc}.

The second part of this paper deals with flexibility of decoders with respect to codeword length.
Sections~\ref{sec:hw-dec} and \ref{sec:sw-dec} discuss such flexibility with respect to hardware and software implementations of the state-of-the-art fast simplified successive-cancellation (Fast-SSC) decoding algorithm, respectively.
The rate and length flexible hardware implementations we present have the same latency and throughput as their rate-only flexible counterparts and incur only a minor increase in complexity. The proposed flexible software decoders can achieve 73\% the throughput of the code-specific decoders.

We would like to mention that some of the proofs presented in this paper were arrived at independently by Li et al. in \cite{Li2015a} and \cite{Li2015b}. Specifically, the result that is most relevant to our setting in \cite{Li2015a} is Theorem 1 thereof as well as the two corollaries that follow. The closest analog in our paper to these results is what one can deduce by combining equations \eqref{eq:FijAndDomination}, \eqref{eq:binaryDominationImpliesUpgradation}, and \eqref{eq:upgradedImpliesBetterBhattacharyya}. However, in contrast to \cite{Li2015a}, our proof is more general since we do not limit ourselves to constructing the polar code via the Bhattacharyya parameter. Also, the results of \cite{Li2015a} are \emph{not} used in that paper for efficient systematic encoding. A systematic encoder based on these results was given later in \cite{Li2015b}, although that encoder is not as amenable to flexible parallel implementation as the encoder proposed in \cite{Sarkis2014} and this paper. We also note that Proposition 3 of \cite{Li2015b} is analogous to our Theorem~\ref{theo:dominationContiguousImpliesInvolution}, although the proofs are different.

\section{Background}
\label{sec:background}
We start by defining what we mean by a ``systematic encoder'', with respect to a general linear code. For integers $0 < k \leq n$, let $G = G_{k \times n}$ denote a $k \times n$ binary matrix with rank $k$. The notation $G$ is used to denote a generator matrix. Namely, the code under consideration is
\[
\myspan(G) = \myset{ \bfv \cdot G \mid \bfv \in \GF(2)^k} \; .
\]

An encoder
\[
\encoder \colon \GF(2)^k \to \myspan(G)
\]
is a one-to-one function mapping an \emph{information bit vector}
\[
\bfu=(u_0,u_1,\ldots,u_{k-1}) \in \GF(2)^k
\]
to a \emph{codeword}
\[
\bfx=(x_0,x_1,\ldots,x_{n-1}) \in \myspan(G) \; .
\]
All the encoders discussed in this paper are linear. Namely, all can be written in the form 
\begin{equation}
\label{eq:generalDecoder}
\encoder(\bfu) = \bfu \cdot \Pi \cdot G \; ,
\end{equation}
where $\Pi = \Pi_{k \times k}$ is an invertible matrix defined over $\GF(2)$.

The encoder $\encoder$ is systematic if there exists a set of $k$ \emph{systematic indices}
\begin{equation}
  \label{eq:Sdefinition}
  S = \mysett{s_j}_{j=0}^{k-1} \; , \quad 0 \leq s_0 < s_1 < \cdots < s_{k-1} \leq n-1 \; ,
\end{equation}
such that restricting $\encoder(\bfu)$ to the indices $S$ yields $\bfu$. Specifically, position $s_i$ of $\bfx$ must contain $u_i$. Note that our definition of ``systematic'' is stronger than some definitions. That is, apart from requiring that the information bits be embedded in the codeword, we further require that the embedding is in the natural order: $u_i$ appears before $u_j$ if $i<j$.

Since $G$ has rank $k$, there exist $k$ linearly independent columns in $G$. Thus, we might naively take $\Pi$ as the inverse of these columns, take $S$ as the indices corresponding to these columns, and state that we are done. Of course, the point of \cite{Sarkis2014} and \cite{Arikan2011} is to show that the calculations involved can be carried out efficiently with respect to the computational model considered. We now briefly present and discuss these two solutions. 

\subsection{The \Arikan\ systematic encoder \cite{Arikan2011}}
Recall \cite{Arikan2009} that a generator matrix of a polar code is obtained as follows. We define the \Arikan\ kernel matrix as
\begin{equation}
\label{eq:F1}
F = \begin{bmatrix}1 & 0\\1 & 1\end{bmatrix} \; .
\end{equation}
The $m$-th Kronecker product of $F$ is denoted $F^{\otimes m}$ and is defined recursively as
\begin{equation}
  \label{eq:FRecursiveDefintion}
  F^{\otimes m} =
  \begin{bmatrix}
    F^{\otimes (m-1)} & 0 \\
    F^{\otimes (m-1)} & F^{\otimes (m-1)} 
  \end{bmatrix} \; , \quad \mbox{where $F^{\otimes 1} = F$} \; .
\end{equation}

From this point forward, we adopt the shorthand 
\[
m \triangleq \log_2 n \; .
\]
In order to construct a polar code of length $n=2^m$, we apply a bit-reversing operation \cite{Arikan2009} to the columns of $F^{\otimes m}$. From the resulting matrix, we erase the $n-k$ rows corresponding to the frozen indices. The resulting $k \times n$ matrix is the generator matrix.

A closely related variant is a code for which the column bit-reversing operation is not carried out. We follow \cite{Arikan2011} and present the encoder there in the context of a non-reversed polar code. Let
the complement of the frozen index set be denoted by
\begin{equation}
  \label{eq:ANotation}
  A = \mysett{\alpha_j}_{j=0}^{k-1} \; , \quad 0 \leq \alpha_0 < \alpha_1 < \cdots < \alpha_{k-1} \leq n-1 \; .
\end{equation}
The set $A$ is termed the set of \emph{active rows}.

A simple observation which is key is that the matrix $F^{\otimes m}$ is lower triangular with all diagonal entries equal to $1$. This is easily proved by induction using the definition of $F$ and (\ref{eq:FRecursiveDefintion}). An immediate corollary is the following. Suppose we start with $F^{\otimes m}$ and keep only the rows indexed by $A$ (thus obtaining the generator matrix $G$). From this matrix, we keep only the $k$ columns indexed by $A$. We are left with a $k \times k$ lower triangular matrix with all diagonal entries equal to $1$. Specifically, we are left with an invertible matrix. In the setting of (\ref{eq:generalDecoder}) and (\ref{eq:Sdefinition}) we have that $\Pi$ is the inverse of this matrix and the set $S$ of systematic indices simply equals $A$.

As previously mentioned, the above description is not enough: we must show an efficient implementation. We now briefly outline the implementation in \cite{Arikan2011}, which results in an encoding algorithm running in time $O(n \cdot \log n)$. Let us recall our goal, we must find a codeword $\bfx=(x_0,x_1,\ldots,x_{n-1})$ such that, using the notation in (\ref{eq:ANotation}), we have that  $x_{\alpha_i}$ equals $u_i$. Since $\bfx$ is a codeword, it is the result of multiplying the generator matrix $G$ by some length-$k$ vector from the left. As mentioned, $G$ is obtained by removing from $F^{\otimes m}$ the rows whose index is not contained in $A$. Thus, we can alternatively state our goal as follows. We must find a codeword $\bfx$ as described above such that $\bfx = \bfv \cdot F^{\otimes m}$, where $\bfv = (v_0,v_1,\ldots,v_{n-1})$ is such that $v_i = 0$ whenever $i \not\in A$.

We will now show a recursive implementation to the systematic encoding function $\Encode_m(\bfu,A)$. Let us start by considering the stopping condition, $m=0$. As a preliminary step, define $F^{\otimes 0} = 1$, a $1 \times 1$ matrix. Note that this definition is consistent with (\ref{eq:F1}) and (\ref{eq:FRecursiveDefintion}) for $m=1$. Next, note that if $m=0$, then the problem is trivial: if $A$ is empty than we are forced to have $\bfv = (0)$, and thus $\Encode_0(\bfu,A)$ returns $\bfx = (0)$, the all-zero codeword. Otherwise, $A=\{0\}$ and we simply take $\bfv = (u_0)$ and return $\bfx = (u_0)$, as prescribed.

Let us now consider the recursion itself. Assume $m \geq 1$ and write $\bfv = (\bfv',\bfv'')$, where $\bfv'$ consists of the first $n/2$ entries of $\bfv$ and $\bfv''$ consists of the last $n/2$ entries of $\bfv$. Let $\bfx = (\bfx',\bfx'')$ be defined similarly. By the block structure in (\ref{eq:FRecursiveDefintion}), we have that
\begin{IEEEeqnarray}{rCl}
\bfx'' &=& \bfv'' \cdot F^{\otimes (m-1)} \label{eq:bfx''}\\
\bfx' &=& \bfv' \cdot F^{\otimes (m-1)} + \bfx'' \label{eq:bfx'}\; .
\end{IEEEeqnarray}
We will find $\bfx$ by first finding $\bfx''$ and then finding $\bfx'$. Towards that end, let
\begin{align}
A' &= \{ \alpha : \alpha \in A \;\; \mbox{and} \;\; \alpha < n/2\},\nonumber\\
A''&= \{ \alpha - n/2 : \alpha \in A \;\; \mbox{and} \;\; \alpha \geq n/2\}.\nonumber
\end{align}
Finding $\bfx''$ is a straightforward recursive process. Namely, by (\ref{eq:bfx''}) if we define $\bfu'' = (u_{i+n/2})_{i \in A''}$, then $\bfx''=\Encode_{m-1}(\bfu'',A'')$. Now, with $\bfx''$ calculated, we can find $\bfx'$. Namely, considering (\ref{eq:bfx'}), we need a $\bfv'$ for which entry $\alpha_i$ of $\bfv' \cdot F^{\otimes (m-1)}$ equals $u_i + x''_{\alpha_i}$. Thus, defining $\bfu' = (u_i + x''_{\alpha_i})_{i \in A'}$, we have that $\bfx'=\Encode_{m-1}(\bfu',A')$.

The main point we want to stress about the above encoder is the \emph{serial} nature of it: \emph{first} calculate $\bfx''$ and \emph{only after} that is done, calculate $\bfx'$.

A parallel, higher complexity implementation of the algorithm in \cite{Arikan2011}, when the frozen bits are set to $0$, can calculate the parity bits directly using matrix multiplication:
\[
\bfx_{A^c} = \bfu (F^{\otimes m}_{AA})^{-1} F^{\otimes m}_{AA^c},
\]
where $F^{\otimes m}_{AA}$ is a sub-matrix containing rows and columns of $F^{\otimes m}$ corresponding to information-bit indices. Similarly, $F^{\otimes m}_{AA^c}$ contains the rows and columns of $F^{\otimes m}$ that correspond to information and frozen bit indices, respectively. The dimensions of $F^{\otimes m}_{AA}$ and $F^{\otimes m}_{AA^c}$ change with code rate, in contrast to our encoder which always uses the fixed $F^{\otimes m}$. A parallel multiplier that can accommodate matrices of varying dimension leads to a significant increase in implementation complexity.

\subsection{The systematic encoder \cite{Sarkis2014}}
\label{subsec:systematicEncoderDescription}
We now give a high-level description of the encoder in \cite{Sarkis2014}. As before, we consider a non-reversed setting. Recall that $A$ in (\ref{eq:ANotation}) is the set of active row indices.
\begin{enumerate}
\item We first expand $\bfu = (u_0,u_1,\ldots,u_{k-1})$ into a vector $\bfv_\mathrm{I}$ of length $n$ as follows: for all $0 \leq i < k$ we set entry $\alpha_i$ of $\bfv_\mathrm{I}$ equal to $u_i$. The remaining $n-k$ entries of $\bfv_\mathrm{I}$ are set to $0$. \label{it:expand}
\item We calculate $\bfv_\mathrm{II} = \bfv_\mathrm{I} \cdot F^{\otimes m}$.
\item The vector $\bfv_\mathrm{III}$ is gotten from $\bfv_\mathrm{II}$ by setting all entries not in $A$ to zero. \label{it:zeroout}
\item We return $\bfx = \bfv_\mathrm{III} \cdot F^{\otimes m}$.
\end{enumerate}
Clearly, steps \ref{it:expand} and \ref{it:zeroout} can be implemented very efficiently in any computational model. The interesting part is calculations of the form $\bfv \cdot F^{\otimes m}$, for a vector $\bfv$ of length $n$.

\begin{figure}
  \centering
  \newcommand{\ubit}[1]{$u_{#1}$}
\newcommand{\fbit}[1]{\color{gray}$u_{#1}$}
\newcommand{\ucw}[1]{$x_{#1}$}
\newcommand{\fcw}[1]{\color{gray}$x_{#1}$}
\newcommand{\ub}[1]{$#1$}
\newcommand{\fb}[1]{\color{gray}$#1$}

\begin{tikzpicture}[scale=0.8,every node/.style={scale=0.8}]

\usetikzlibrary{shapes,positioning,arrows,decorations.markings,fit}

\definecolor{varnode_fill}{RGB}{0,0,0}
\definecolor{chknode_fill}{RGB}{255,255,255}

\tikzset{
  chk/.style={draw,fill=chknode_fill,circle,minimum size=0.3cm, inner sep=0},
  var/.style={draw,fill=varnode_fill,circle,minimum size=0.1cm, inner sep=0},
  sep/.style={rectangle,minimum width=0.25cm, inner sep=0},
  empty/.style={rectangle, inner sep=0},
  bit/.style={circle, inner sep = 0}
}

\matrix[row sep=1mm, column sep=1mm] {
  \node[bit] (n0s0) {\fb{0}};    & \node[sep] {}; & \node[chk] (n0s1) {\scriptsize$+$};  & \node[sep] {}; &                                     & \node[chk] (n0s2) {\scriptsize$+$}; & \node[sep] {}; &                                     &                                     &                                     & \node[chk] (n0s3) {\scriptsize$+$}; & \node[sep] {}; & \node[bit] (xn0s0) {\fb{0}};   & \node[sep] {}; & \node[chk] (xn0s1) {\scriptsize$+$};  & \node[sep] {}; &                                      & \node[chk] (xn0s2) {\scriptsize$+$}; & \node[sep] {}; &                                      &                                      &                                      & \node[chk] (xn0s3) {\scriptsize$+$}; & \node[sep] {}; & \node[bit] (xn0s4) {$x_0$};\\
  \node[bit] (n1s0) {\fb{0}};    & \node[sep] {}; & \node[var] (n1s1) {};                & \node[sep] {}; & \node[chk] (n1s2) {\scriptsize$+$}; &                                     & \node[sep] {}; &                                     &                                     & \node[chk] (n1s3) {\scriptsize$+$}; &                                     & \node[sep] {}; & \node[bit] (xn1s0) {\fb{0}};   & \node[sep] {}; & \node[var] (xn1s1) {};                & \node[sep] {}; & \node[chk] (xn1s2) {\scriptsize$+$}; &                                      & \node[sep] {}; &                                      &                                      & \node[chk] (xn1s3) {\scriptsize$+$}; &                                      & \node[sep] {}; & \node[bit] (xn1s4) {$x_1$};\\
  \node[bit] (n2s0) {\fb{0}};    & \node[sep] {}; & \node[chk] (n2s1) {\scriptsize$+$};  & \node[sep] {}; &                                     & \node[var] (n2s2) {};               & \node[sep] {}; &                                     & \node[chk] (n2s3) {\scriptsize$+$}; &                                     &                                     & \node[sep] {}; & \node[bit] (xn2s0) {\fb{0}};   & \node[sep] {}; & \node[chk] (xn2s1) {\scriptsize$+$};  & \node[sep] {}; &                                      & \node[var] (xn2s2) {};               & \node[sep] {}; &                                      & \node[chk] (xn2s3) {\scriptsize$+$}; &                                      &                                      & \node[sep] {}; & \node[bit] (xn2s4) {$x_2$};\\
  \node[bit] (n3s0) {\ub{u_0}};  & \node[sep] {}; & \node[var] (n3s1) {};                & \node[sep] {}; & \node[var] (n3s2) {};               &                                     & \node[sep] {}; & \node[chk] (n3s3) {\scriptsize$+$}; &                                     &                                     &                                     & \node[sep] {}; & \node[bit] (xn3s0) {};         & \node[sep] {}; & \node[var] (xn3s1) {};                & \node[sep] {}; & \node[var] (xn3s2) {};               &                                      & \node[sep] {}; & \node[chk] (xn3s3) {\scriptsize$+$}; &                                      &                                      &                                      & \node[sep] {}; & \node[bit] (xn3s4) {$u_0$};\\
  \node[bit] (n4s0) {\ub{u_1}};  & \node[sep] {}; & \node[chk] (n4s1) {\scriptsize$+$};  & \node[sep] {}; &                                     & \node[chk] (n4s2) {\scriptsize$+$}; & \node[sep] {}; &                                     &                                     &                                     & \node[var] (n4s3) {};               & \node[sep] {}; & \node[bit] (xn4s0) {};         & \node[sep] {}; & \node[chk] (xn4s1) {\scriptsize$+$};  & \node[sep] {}; &                                      & \node[chk] (xn4s2) {\scriptsize$+$}; & \node[sep] {}; &                                      &                                      &                                      & \node[var] (xn4s3) {};               & \node[sep] {}; & \node[bit] (xn4s4) {$u_1$};\\
  \node[bit] (n5s0) {\ub{u_2}};  & \node[sep] {}; & \node[var] (n5s1) {};                & \node[sep] {}; & \node[chk] (n5s2) {\scriptsize$+$}; &                                     & \node[sep] {}; &                                     &                                     & \node[var] (n5s3) {};               &                                     & \node[sep] {}; & \node[bit] (xn5s0) {};         & \node[sep] {}; & \node[var] (xn5s1) {};                & \node[sep] {}; & \node[chk] (xn5s2) {\scriptsize$+$}; &                                      & \node[sep] {}; &                                      &                                      & \node[var] (xn5s3) {};               &                                      & \node[sep] {}; & \node[bit] (xn5s4) {$u_2$};\\
  \node[bit] (n6s0) {\ub{u_3}};  & \node[sep] {}; & \node[chk] (n6s1) {\scriptsize$+$};  & \node[sep] {}; &                                     & \node[var] (n6s2) {};               & \node[sep] {}; &                                     & \node[var] (n6s3) {};               &                                     &                                     & \node[sep] {}; & \node[bit] (xn6s0) {};         & \node[sep] {}; & \node[chk] (xn6s1) {\scriptsize$+$};  & \node[sep] {}; &                                      & \node[var] (xn6s2) {};               & \node[sep] {}; &                                      & \node[var] (xn6s3) {};               &                                      &                                      & \node[sep] {}; & \node[bit] (xn6s4) {$u_3$};\\
  \node[bit] (n7s0) {\ub{u_4}};  & \node[sep] {}; & \node[var] (n7s1) {};                & \node[sep] {}; & \node[var] (n7s2) {};               &                                     & \node[sep] {}; & \node[var] (n7s3) {};               &                                     &                                     &                                     & \node[sep] {}; & \node[bit] (xn7s0) {};         & \node[sep] {}; & \node[var] (xn7s1) {};                & \node[sep] {}; & \node[var] (xn7s2) {};               &                                      & \node[sep] {}; & \node[var] (xn7s3) {};               &                                      &                                      &                                      & \node[sep] {}; & \node[bit] (xn7s4) {$u_4$};\\
};

\path[-] (n0s0) edge (n0s1) (n0s1) edge (n0s2) (n0s2) edge (n0s3) (n0s3) edge (xn0s0);
\path[-] (n1s0) edge (n1s1) (n1s1) edge (n1s2) (n1s2) edge (n1s3) (n1s3) edge (xn1s0);
\path[-] (n2s0) edge (n2s1) (n2s1) edge (n2s2) (n2s2) edge (n2s3) (n2s3) edge (xn2s0);
\path[-] (n3s0) edge (n3s1) (n3s1) edge (n3s2) (n3s2) edge (n3s3) (n3s3) edge (xn3s1);
\path[-] (n4s0) edge (n4s1) (n4s1) edge (n4s2) (n4s2) edge (n4s3) (n4s3) edge (xn4s1);
\path[-] (n5s0) edge (n5s1) (n5s1) edge (n5s2) (n5s2) edge (n5s3) (n5s3) edge (xn5s1);
\path[-] (n6s0) edge (n6s1) (n6s1) edge (n6s2) (n6s2) edge (n6s3) (n6s3) edge (xn6s1);
\path[-] (n7s0) edge (n7s1) (n7s1) edge (n7s2) (n7s2) edge (n7s3) (n7s3) edge (xn7s1);

\path[-] (n0s1) edge (n1s1);
\path[-] (n2s1) edge (n3s1);
\path[-] (n4s1) edge (n5s1);
\path[-] (n6s1) edge (n7s1);

\path[-] (n0s2) edge (n2s2);
\path[-] (n1s2) edge (n3s2);
\path[-] (n4s2) edge (n6s2);
\path[-] (n5s2) edge (n7s2);

\path[-] (n0s3) edge (n4s3);
\path[-] (n1s3) edge (n5s3);
\path[-] (n2s3) edge (n6s3);
\path[-] (n3s3) edge (n7s3);


\path[-] (xn0s0) edge (xn0s1) (xn0s1) edge (xn0s2) (xn0s2) edge (xn0s3) (xn0s3) edge (xn0s4);
\path[-] (xn1s0) edge (xn1s1) (xn1s1) edge (xn1s2) (xn1s2) edge (xn1s3) (xn1s3) edge (xn1s4);
\path[-] (xn2s0) edge (xn2s1) (xn2s1) edge (xn2s2) (xn2s2) edge (xn2s3) (xn2s3) edge (xn2s4);
\path[-] (xn3s0) edge (xn3s1) (xn3s1) edge (xn3s2) (xn3s2) edge (xn3s3) (xn3s3) edge (xn3s4);
\path[-] (xn4s0) edge (xn4s1) (xn4s1) edge (xn4s2) (xn4s2) edge (xn4s3) (xn4s3) edge (xn4s4);
\path[-] (xn5s0) edge (xn5s1) (xn5s1) edge (xn5s2) (xn5s2) edge (xn5s3) (xn5s3) edge (xn5s4);
\path[-] (xn6s0) edge (xn6s1) (xn6s1) edge (xn6s2) (xn6s2) edge (xn6s3) (xn6s3) edge (xn6s4);
\path[-] (xn7s0) edge (xn7s1) (xn7s1) edge (xn7s2) (xn7s2) edge (xn7s3) (xn7s3) edge (xn7s4);

\path[-] (xn0s1) edge (xn1s1);
\path[-] (xn2s1) edge (xn3s1);
\path[-] (xn4s1) edge (xn5s1);
\path[-] (xn6s1) edge (xn7s1);

\path[-] (xn0s2) edge (xn2s2);
\path[-] (xn1s2) edge (xn3s2);
\path[-] (xn4s2) edge (xn6s2);
\path[-] (xn5s2) edge (xn7s2);

\path[-] (xn0s3) edge (xn4s3);
\path[-] (xn1s3) edge (xn5s3);
\path[-] (xn2s3) edge (xn6s3);
\path[-] (xn3s3) edge (xn7s3);

\end{tikzpicture}
  \caption{The systematic encoder of \cite{Sarkis2014} for an (8, 5) polar code.}
  \label{fig:enc}
\end{figure}
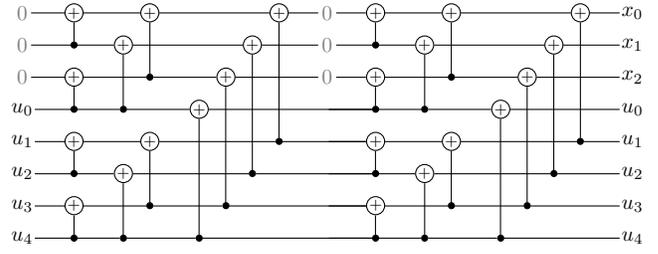

As we will expand on later, the main merit of \cite{Sarkis2014} is that the computation of $\bfv \cdot F^{\otimes m}$ can be done \emph{in parallel}. Namely, if $\bfv = (\bfv',\bfv'')$, where $\bfv'$ (respectively, $\bfv''$) equals the first (respectively, last) $n/2$ entries of $\bfv$, then one can calculate $\bfv' \cdot F^{\otimes (m-1)}$ and $\bfv'' \cdot F^{\otimes (m-1)}$ \emph{concurrently} and then, by (\ref{eq:FRecursiveDefintion}), combine the results  to get 
\[
\bfv \cdot F^{\otimes m} = ([\bfv' \cdot F^{\otimes (m-1)}] + [\bfv'' \cdot F^{\otimes (m-1)}], [\bfv'' \cdot F^{\otimes (m-1)}]) \; .
\]

We also note that the systematic encoder in \cite{Sarkis2014} is easily described as two applications of a non-systematic encoder, with a zeroing operation applied in-between. Thus, any advances made with respect to non-systematic encoding of polar codes immediately yield advances in systematic encoding.

Lastly, we state that both the encoder presented in \cite{Arikan2011} as well as the one presented in \cite{Sarkis2014} produce the same codeword when given the same information vector. To see this, note that on the one hand, both encoders operate with respect to the same code of dimension $k$. That is, with respect to the same generator matrix $G$ described above. On the other hand, by definition, both encoders produce the same output when restricted to the $k$ systematic indexes $A$ of the codeword. That is, to $k$ indexes such that restricting the generator matrix $G$ to them results in a $k \times k$ invertible matrix, as previously explained. Thus, the error-correction performance of a system utilizing the same decoder with either encoder remains the same

\section{Systematic, reversed, and non-reversed codes}
\label{sec:systematicEncodingDefs}
This section is devoted to recasting the concepts and operation presented in the previous section into matrix terminology. We start by discussing a general linear code, and then specialize to both non-reversed and reversed polar codes. Recalling the definition of $S$ as the set of systematic indices, define the \emph{restriction matrix} $R = R_{n \times k}$ corresponding to $S$ as 
\begin{equation}
  \label{eq:Rdefinition}
  R=(R_{i,j})_{i=0}^{n-1}{}_{j=0}^{k-1} \; , \quad \mbox{where} \quad 
  R_{i,j} =
  \begin{cases}
    1 & \mbox{if $i = s_j$} \; , \\
    0 & \mbox{otherwise} \; .
  \end{cases}
\end{equation}
With this definition at hand, we require that a systematic encoder satisfy $\encoder(\bfu) \cdot R = \bfu$, or equivalently that
\begin{equation}
  \label{eq:matrixDefinitionOfSystematic}
  \Pi \cdot G \cdot R = I \; ,
\end{equation}
where $I$ above denotes the $k \times k$ identity matrix. Our proofs will center on showing that (\ref{eq:matrixDefinitionOfSystematic}) holds. 

\subsection{Non-reversed polar codes}
In this subsection, we consider a non-reversed polar code. Recall the definition in (\ref{eq:ANotation}) of $A$ being the set of active rows, where the $j$th smallest element of $A$ is denoted $\alpha_j$. For this case, recall that we define $S$ as equal to $A$ and $s_j$ as equal to $\alpha_j$.

Define the matrix $E$ as
\begin{equation}
  \label{eq:Edefinition}
  E=(E_{i,j})_{i=0}^{k-1}{}_{j=0}^{n-1} \; , \quad \mbox{where} \quad 
  E_{i,j} =
  \begin{cases}
    1 & \mbox{if $j = \alpha_i$} \; , \\
    0 & \mbox{otherwise} \; .
  \end{cases}
\end{equation}
The matrix $E$ will be useful in several respects. First, note by the above that applying $E$ to the left of a matrix with $n$ rows results in a submatrix containing only the rows indexed by $A$.
Thus, we have that
\begin{equation}
  \label{eq:GNonReversedDef}
  \GNonReversed = E \cdot F^{\otimes m} \; ,
\end{equation}
where $\GNonReversed$ is the generator matrix of our code, and ``nrv'' is short for ``non-reversed''.

Next, note that by applying $E$ to the right of a vector $\bfu$ of length $k$, we manufacture a vector $\bfv_\mathrm{I}$ such that the entries indexed by $\alpha_j$ equal $u_j$ and all other entries equal zero. That is,
\[
\bfv_\mathrm{I} = \bfu \cdot E \; ,
\]
as per step \ref{it:expand} of the algorithm described in Subsection~\ref{subsec:systematicEncoderDescription}. Because of this property, we refer to $E$ as the \emph{expanding matrix}. 

Let us move on to step \ref{it:zeroout} of the algorithm. Simple algebra yields that
\[
\bfv_\mathrm{III} = \bfv_\mathrm{II} \cdot E^T \cdot E \; .
\]
That is, multiplying a vector of length $n$ from the right by $E^T \cdot E$ results in a vector in which the entries indexed by $A$ remain the same while the entries not indexed by $A$ are set to zero.

The above equations yield a succinct description of our algorithm,
\begin{equation}
  \label{eq:encoderNonReversedDef}
  \encoderNonReversed(\bfu) = \bfu \cdot \underbrace{ E \cdot F^{\otimes m} \cdot E^T}_{\Pi} \cdot \underbrace{E \cdot F^{\otimes m}}_{\GNonReversed} \; .
\end{equation}
We end this section by noting that by (\ref{eq:Rdefinition}) and (\ref{eq:Edefinition}), we have that
\[
E^T = R \; .
\]
Thus, recalling (\ref{eq:matrixDefinitionOfSystematic}), our aim is to prove that  
\begin{equation}
  \label{eq:nonreversedEncoderInvolution}
  E \cdot F^{\otimes m} \cdot E^T \cdot E \cdot  F^{\otimes m} \cdot E^T = I \; .
\end{equation}
Showing this will further imply that the corresponding $\Pi$ in (\ref{eq:encoderNonReversedDef}) is indeed invertible.

\subsection{Reversed polar codes}
As explained, we will consider bit-reversed as well as non-bit-reversed polar codes. Let us introduce corresponding notation.
For an integer $0 \leq i < n$, denote the binary representation of $i$ as
\ifCLASSOPTIONtwocolumn
\begin{multline}
  \label{eq:binaryRepresentation}
  \binary{i} = (i_0,i_1,\ldots,i_{m-1}) \; , \\
  \mbox{where}  \; i = \sum_{j=0}^{m-1} i_j 2^j \; \mbox{and} \;  i_j \in \{0,1\} \; .
\end{multline}
\else
\begin{equation}
  \label{eq:binaryRepresentation}
  \binary{i} = (i_0,i_1,\ldots,i_{m-1}) \; ,
  \text{where}  \; i = \sum_{j=0}^{m-1} i_j 2^j \; \text{and} \;  i_j \in \{0,1\} \; .
\end{equation}
\fi
For $i$ as above, we define $\reverse{i}$ as the integer with reversed binary representation. That is,
\[
\binary{\reverse{i}} =  (i_{m-1},i_{m-2},\ldots,i_0)\; , \quad \reverse{i} = \sum_{j=0}^{m-1} i_j 2^{m-1-j} \; .
\]

As in \cite{Arikan2009}, we denote the $n \times n$ bit reversal matrix as $B_n$. Recall that $B_n$ is a permutation matrix. Specifically, multiplying a matrix from the left (right) by $B_n$ results in a matrix in which row (column) $i$ equals row (column) $\reverse{i}$ of the original matrix. 

Recall that we have denoted by $A$ the set of active rows. We stress that this notation holds for \emph{both the reversed as well as the non-reversed setting}. Thus, recalling (\ref{eq:Edefinition}), we have analogously to (\ref{eq:GNonReversedDef}) that
\[
\GReversed = E  \cdot F^{\otimes m} \cdot B_n
\]
where $\GReversed$ is the generator matrix of our code, and ``rv'' is short for ``reversed''. By \cite[Proposition 16]{Arikan2009}, we know that $B_n \cdot F^{\otimes m} = F^{\otimes m} \cdot B_n$. Thus, it also holds that
\begin{equation}
  \label{eq:GReversedDef}
  \GReversed = E \cdot B_n \cdot F^{\otimes m} \; .
\end{equation}

In the interest of a lighter notation later on, we now ``fold'' the bit-reversing operation into the set $A$. Thus,
define the set of bit-reversed active rows, $\reverse{A}$, gotten from the set of active rows $A$ by applying the bit-reverse operation on each element $\alpha_i$. As before, we order the elements of $\reverse{A}$ in increasing order and denote
\begin{equation}
  \label{eq:AReversedNotation}
  \reverse{A} = \mysett{\beta_j}_{j=0}^{k-1} \; , \quad 0 \leq \beta_0 < \beta_1< \cdots < \beta_{k-1} \leq n-1 \; .
\end{equation}
Recall that the expansion matrix $E$ was defined using $A$. We now define $\reverse{E} = \reverse{E}_{k \times n}$ according to $\reverse{A}$ in exactly the same way. That is,
\begin{equation}
  \label{eq:EReversedDefinition}
  \reverse{E}=(\reverse{E}_{i,j})_{i=0}^{k-1}{}_{j=0}^{n-1} \; , \quad\!\! \mbox{where} \quad 
  \reverse{E}_{i,j} =
  \begin{cases}
    1 & \mbox{if $j = \beta_i$} \; , \\
    0 & \mbox{otherwise} \; .
  \end{cases}
\end{equation}
Note that $E \cdot B$ and $\reverse{E}$ are the same, up to a permutation of rows (for $i$ fixed, the reverse of $\alpha_i$ does not generally equal $\beta_i$, hence the need for a permutation). Thus, by (\ref{eq:GReversedDef}),
\begin{equation}
  \label{eq:GReversedPrimeDef}
  \GReversed' = \reverse{E} \cdot F^{\otimes m}
\end{equation}
is a generator matrix spanning the same code as $\GReversed$. Analogously to (\ref{eq:encoderNonReversedDef}), our encoder for the reversed code is given by
\begin{equation}
  \label{eq:encoderReversedDef}
  \encoderReversed(\bfu) = \bfu \cdot \underbrace{ \reverse{E} \cdot F^{\otimes m} \cdot (\reverse{E})^T}_{\Pi} \cdot \underbrace{\reverse{E} \cdot F^{\otimes m}}_{\GReversed'} \; .
\end{equation}
We now highlight the similarities and differences with respect to the non-reversed encoder. First, note that for the reversed encoder, the set of systematic indices is $\reverse{A}$, as opposed to $A$ for the non-reversed encoder. Apart from that, everything remains the same. Namely, conceptually, we are simply operating the non-reversed encoder with $\reverse{A}$ in place of $A$. Specifically, note that as in the non-reversed case, the encoder produces a codeword such that the information bits are embedded in the natural order.

Analogously to (\ref{eq:nonreversedEncoderInvolution}), our aim is to prove that
\begin{equation}
  \label{eq:reversedEncoderInvolution}
  \reverse{E} \cdot F^{\otimes m} \cdot (\reverse{E})^T \cdot \reverse{E} \cdot  F^{\otimes m} \cdot (\reverse{E})^T = I \; .
\end{equation}

\section{Domination contiguity implies involution}
\label{sec:dominationContiguiryImpliesInvolution}
In this section we prove that our encoders are valid by proving that (\ref{eq:nonreversedEncoderInvolution}) and (\ref{eq:reversedEncoderInvolution}) indeed hold. A square matrix is called an \emph{involution} if multiplying the matrix by itself yields the identity matrix. With this terminology at hand, we must prove that both $E \cdot F^{\otimes m} \cdot E^T$ and $\reverse{E} \cdot F^{\otimes m} \cdot (\reverse{E})^T$ are involutions.

Interestingly, and in contrast with the original systematic encoder presented in \cite{Arikan2011}, the proof of correctness centers on the structure of $A$. That is, in \cite{Arikan2011}, any set of $k$ active (non-frozen) channels has a corresponding systematic encoder. In contrast, consider as an example the case in which $n=4$ and $A=\{0,1,3\}$. By our definitions,
\[
E =  
\left[
  \begin{smallmatrix}
    1 & 0 & 0 & 0 \\
    0 & 1 & 0 & 0 \\
    0 & 0 & 0 & 1
  \end{smallmatrix}
\right]
\; , \quad 
E^T = \left[
  \begin{smallmatrix}
    1 & 0 & 0 \\
    0 & 1 & 0 \\
    0 & 0 & 0 \\
    0 & 0 & 1
  \end{smallmatrix}
\right]
\; , \quad \mbox{and} \quad
F^{\otimes 2} = \left[
  \begin{smallmatrix}
    1 & 0 & 0 & 0 \\
    1 & 1 & 0 & 0 \\
    1 & 0 & 1 & 0 \\
    1 & 1 & 1 & 1 
  \end{smallmatrix}
\right] \; .
\]
Thus,
\ifCLASSOPTIONtwocolumn
\begin{multline*}
  E \cdot F^{\otimes 2} \cdot E^T = 
  \left[
    \begin{smallmatrix}
      1 & 0 & 0 \\
      1 & 1 & 0 \\
      1 & 1 & 1 
    \end{smallmatrix}
  \right]
  \; , \quad \mbox{and}\\
  (E \cdot F^{\otimes 2} \cdot E^T)\cdot (E \cdot F^{\otimes 2} \cdot E^T) = 
  \left[
    \begin{smallmatrix}
      1 & 0 & 0 \\
      0 & 1 & 0 \\
      1 & 0 & 1 
    \end{smallmatrix}
  \right] \; .
\end{multline*}
\else
\[
  E \cdot F^{\otimes 2} \cdot E^T = 
  \left[
    \begin{smallmatrix}
      1 & 0 & 0 \\
      1 & 1 & 0 \\
      1 & 1 & 1 
    \end{smallmatrix}
  \right]
  \; , \text{ and }
  (E \cdot F^{\otimes 2} \cdot E^T)\cdot (E \cdot F^{\otimes 2} \cdot E^T) = 
  \left[
    \begin{smallmatrix}
      1 & 0 & 0 \\
      0 & 1 & 0 \\
      1 & 0 & 1 
    \end{smallmatrix}
  \right] \; .
\]
\fi
Note that the rightmost matrix above is \emph{not} an identity matrix. A similar calculation shows that $\reverse{E} \cdot F^{\otimes 2} \cdot (\reverse{E})^T$ is not an involution either.

The apparent contradiction to the correctness of our algorithms will be rectified in the next section. In brief, using terminology defined in Section~\ref{sec:polarCodesSatisfyDominationContiguity}, the fact that $A=\{0,1,3\}$ implies that $W^{+-}$ is frozen while $W^{--}$ is unfrozen. However, this cannot correspond to a valid polar code since $W^{+-}$ is upgraded with respect to $W^{--}$.

We now characterize the $A$ for which (\ref{eq:nonreversedEncoderInvolution}) and (\ref{eq:reversedEncoderInvolution}) hold. Recall our notation for binary representation given in (\ref{eq:binaryRepresentation}). For $0 \leq i,j \leq n$, denote 
\[
\binary{i} = (i_0,i_1,\ldots,i_{m -1}) \; , \quad \binary{j} = (j_0,j_1,\ldots,j_{m -1}) \; .
\]
We define the \emph{binary domination} relation, denoted $\dominates$,  as follows.
\[
i \dominates j \quad \mbox{iff for all $0 \leq t < m$, we have $i_t \geq j_t$} \; .
\]
Namely, $i \dominates j$ iff the support of $\binary{i}$ (the indices $t$ for which $i_t=1$) contains the support of $\binary{j}$.

We say that a set of indices $A \subseteq \{0,1,\ldots,n-1\}$ is \emph{domination contiguous} if for all $h,j \in A$ and for all $0 \leq i < n$ such that $h \dominates i$ and $i \dominates j$, it holds that $i \in A$. For easy reference:
\begin{equation}
  \label{eq:dominationContiguous}
  (h,j \in A \quad \mbox{and} \quad h \dominates i \dominates j) \then i \in A \; .
\end{equation}
\begin{theorem}
  \label{theo:dominationContiguousImpliesInvolution}
  Let the active rows set $A \subseteq \{0,1,\ldots,n-1\}$ be domination contiguous, as defined in (\ref{eq:dominationContiguous}). Let $E$ and $\reverse{E}$ be defined according to (\ref{eq:ANotation}), (\ref{eq:Edefinition}), (\ref{eq:AReversedNotation}), and (\ref{eq:EReversedDefinition}). Then, $E \cdot F^{\otimes m} \cdot E^T$ and $\reverse{E} \cdot F^{\otimes m} \cdot (\reverse{E})^T$ are involutions. That is, (\ref{eq:nonreversedEncoderInvolution}) and (\ref{eq:reversedEncoderInvolution}) hold.
\end{theorem}
\begin{proof}
  We first note that for $0 \leq i,j < n$, we have that $i \dominates j$ iff $\reverse{i} \dominates \reverse{j}$. Thus, if $A$ is domination contiguous then so is $\reverse{A}$. As a consequence, proving that $E \cdot F^{\otimes m} \cdot E^T$ is an involution will immediately imply that $\reverse{E} \cdot F^{\otimes m} \cdot (\reverse{E})^T$ is an involution as well. Let us prove the former---that is, let us prove (\ref{eq:nonreversedEncoderInvolution}).

  We start by noting a simple characterization of $F^{\otimes m}$, where $F$ is defined as in \eqref{eq:F1}. Namely, the entry at row $i$ and column $j$ of $F^{\otimes m}$ is easily calculated:
  \begin{equation}
    \label{eq:FijAndDomination}
    (F^{\otimes m})_{i,j} = 
    \begin{cases}
      1 & i \dominates j \; ,\\
      0 & \mbox{otherwise} \; .
    \end{cases}
  \end{equation}

  To see this, consider the recursive definition of $F^{\otimes m}$ given in (\ref{eq:FRecursiveDefintion}). Obviously, $(F^{\otimes m})_{i,j}$ equals $0$ if we are at the upper right $(n/2) \times (n/2)$ block. That is, if $i_{m-1}$ (the most-significant bit of $i$) equals $0$ and $j_{m-1}$ equals $1$. Next, consider the other three blocks and note that for them, $i \dominates j$ iff $i \mod 2^{m-1}$ dominates $j \mod 2^{m-1}$. Since the remaining blocks all contain the same matrix, it suffices to prove the claim for the lower left block. Thus, we continue recursively with $i \mod 2^{m-1}$ and $j \mod 2^{m-1}$.

  Recalling (\ref{eq:ANotation}) and the fact that $|A|=k$, we adopt the following shorthand: for $0 \leq p,q,r < k$ given, let
  \[
  h = \alpha_p \; , \quad i = \alpha_q \; , \quad j = \alpha_r \; .
  \]
  By the above, a straightforward derivation yields that
\ifCLASSOPTIONtwocolumn
  \begin{multline*}
    (E \cdot F^{\otimes m} \cdot E^T)_{p,q} = (F^{\otimes m})_{h,i} \\
    \mbox{and} \quad (E \cdot F^{\otimes m} \cdot E^T)_{q,r} = (F^{\otimes m})_{i,j} \; .
  \end{multline*}
\else
\[
    (E \cdot F^{\otimes m} \cdot E^T)_{p,q} = (F^{\otimes m})_{h,i}
    \quad \text{and}\quad (E \cdot F^{\otimes m} \cdot E^T)_{q,r} = (F^{\otimes m})_{i,j} \; .
\]
\fi
  Thus,
\ifCLASSOPTIONtwocolumn
  \begin{multline}
    \label{eq:productOfFSubmatrices}
    \bigg((E \cdot F^{\otimes m} \cdot E^T) \cdot (E \cdot F^{\otimes m} \cdot E^T)\bigg)_{p,r} \\
    = \sum_{q=0}^{k-1}(E \cdot F^{\otimes m} \cdot E^T)_{p,q} \cdot (E \cdot F^{\otimes m} \cdot E^T)_{q,r} \\
    = \sum_{i \in A} (F^{\otimes m})_{h,i} \cdot (F^{\otimes m})_{i,j} \; .
  \end{multline}
\else
  \begin{align}
    \label{eq:productOfFSubmatrices}
    \bigg((E \cdot F^{\otimes m} \cdot E^T) \cdot (E \cdot F^{\otimes m} \cdot E^T)\bigg)_{p,r}
    &= \sum_{q=0}^{k-1}(E \cdot F^{\otimes m} \cdot E^T)_{p,q} \cdot (E \cdot F^{\otimes m} \cdot E^T)_{q,r} \nonumber\\
    &= \sum_{i \in A} (F^{\otimes m})_{h,i} \cdot (F^{\otimes m})_{i,j} \; .
  \end{align}
\fi
  Proving (\ref{eq:nonreversedEncoderInvolution}) is now equivalent to proving that the right-hand side of (\ref{eq:productOfFSubmatrices}) equals $1$ iff $h$ equals $j$. Recalling (\ref{eq:FijAndDomination}), this is equivalent to showing that if $h \neq j$, then there is an even number of $i \in A$ for which
  \begin{equation}
    \label{eq:h_dominates_i_dominates_j}
    h \dominates i \quad \mbox{and} \quad i \dominates j \; ,
  \end{equation}
  while if $h= j$, then there is an odd number of such $i$. 

  We distinguish between 3 cases.
  \begin{enumerate}
  \item If $h = j$, then there is a single $0 \leq i < n$ for which (\ref{eq:h_dominates_i_dominates_j}) holds. Namely, $i = h = j$. Since $h,j \in A$, we have that $i \in A$ as well. Since $1$ is odd, we are finished with the first case.
  \item If $h \neq j$ and $h \not\dominates j$, then there can be no $i$ for which (\ref{eq:h_dominates_i_dominates_j}) holds. Since $0$ is an even integer, we are done with this case as well.
  \item If $h \neq j$ and $h \dominates j$, then the support of the binary vector $\binary{j} = (j_0,j_1,\ldots,j_{m-1})$ is contained in and distinct from the support of the binary vector $\binary{h} = (h_0,h_1,\ldots,h_{m-1})$. A moment of thought reveals that the number of $0 \leq i < n$ for which (\ref{eq:h_dominates_i_dominates_j}) holds is equal to $2^{w(h) - w(j)}$, where $w(h)$ and $w(j)$ represent the support size of $\binary{h}$ and $\binary{j}$, respectively. Since $h \neq j$ and $h \dominates j$, we have that $w(h) - w(j) > 0$. Thus, $2^{w(h) - w(j)}$ is even. Since $h , j \in A$ and $A$ is domination contiguous, all of the above mentioned $i$ are members of $A$. To sum up, an even number of $i \in A$ satisfy (\ref{eq:h_dominates_i_dominates_j}), as required. \qedhere{}
\end{enumerate}
\end{proof}

Recall \cite[Section X]{Arikan2009} that an $(r,m)$ Reed-Muller code has length $n=2^m$ and is formed by taking the set $A$ to contain all indices $i$ such that the support of $\binary{i}$ has size at least $r$. Clearly, such an $A$ is domination contiguous, as defined in (\ref{eq:dominationContiguous}). Hence, the following is an immediate corollary of Theorem~\ref{theo:dominationContiguousImpliesInvolution}, and states that our encoders are valid for Reed-Muller codes.
\begin{corollary}
  Let the active row set $A$ correspond to an $(r,m)$ Reed-Muller code. Let $E$ and $\reverse{E}$ be defined according to (\ref{eq:ANotation}), (\ref{eq:Edefinition}), (\ref{eq:AReversedNotation}), and (\ref{eq:EReversedDefinition}), where $n=2^m$. Then, $E \cdot F^{\otimes m} \cdot E^T$ and $\reverse{E} \cdot F^{\otimes m} \cdot (\reverse{E})^T$ are involutions. That is, (\ref{eq:nonreversedEncoderInvolution}) and (\ref{eq:reversedEncoderInvolution}) hold and thus our two encoders are valid.
\end{corollary}

\section{Polar codes satisfy domination contiguity}
\label{sec:polarCodesSatisfyDominationContiguity}
The previous section concluded with proving that our encoders are valid for Reed-Muller codes. Our aim in this section is to prove that our encoders are valid for polar codes. In order to do so, we first define the concept of a (stochastically) upgraded channel.

A channel $W$ with input alphabet $\calX$ and output alphabet $\calY$ is denoted $W : \calX \to \calY$. The probability of receiving $y \in \calY$ given that $x \in \calX$ was transmitted is denoted $W(y|x)$. Our channels will be binary input, memoryless, and output symmetric (BMS).  Binary: the channel input alphabet will be denoted as $\calX = \{0,1\}$. Memoryless: the probability of receiving the vector $(y_i)_{i=0}^{n-1}$ given that the vector $(x_i)_{i=0}^{n-1}$ was transmitted is $\prod_{i=0}^{n-1} W(y_i|x_i)$. Symmetric: there exists a permutation $\pi : \calY \to \calY$ such that that for all $y \in \calY$, $\pi(\pi(y)) = y$ and $W(y|0) = W(\pi(y)|1)$.

We say that a channel $W : \calX \to \calY$ is upgraded with respect to a channel $Q: \calX \to \calZ$ if there exists a channel $\Phi : \calY \to \calZ$ such that concatenating $\Phi$ to $W$ results in $Q$. Formally, for all $x \in \calX$ and $z \in \calZ$,
\[
Q(z|x) = \sum_{y \in \calY} W(y|x) \cdot \Phi(z|y) \; .
\]
We denote $W$ being upgraded with respect to $Q$ as $W \upgraded Q$. As we will soon see, using the same notation for upgraded channels and binary domination is helpful.

Let $W: \calX \to \calY$ be a BMS channel. 
Let $W^-: \calX \to \calY^2$ and $W^+ : \calX \to \calY^2 \times \calX$ be the ``minus'' and ``plus'' transform as defined in \cite{Arikan2009}. That is,
\begin{align}
  W^-(y_0,y_1|u_0) &= \frac{1}{2} \sum_{u_1 \in \{0,1\}} W(y_0|u_0 + u_1) \cdot W(y_1|u_1) \; , \nonumber\\
  W^+(y_0,y_1,u_0|u_1) &= \frac{1}{2} W(y_0|u_0 + u_1) \cdot W(y_1|u_1) \; . \nonumber
\end{align}
The claim in the following lemma seems to be well known in the community, and is very easy to prove. Still, since we have not found a place in which the proof is stated explicitly, we supply it as well.
\begin{lemma}
\label{lemm:plusUpgradedWithRespectToMinus}
  Let $W: \calX \to \calY$ be a BMS channel. Then, $W^+$ is upgraded with respect to $W^-$, 
  \begin{equation}
    \label{eq:plusUpgradedWithRespectToMinus}
    W^+ \upgraded W^- \; .
  \end{equation}
\end{lemma}
\begin{proof}
  We prove that $W^+ \upgraded W$ and $W \upgraded W^-$. Since ``$\upgraded$'' is easily seen to be a transitive relation, the proof follows. To show that $W^+ \upgraded W$, take $\Phi : \calY^2 \times \calX \to \calY$ as the channel which maps $(y_0,y_1,u_0)$ to $y_1$ with probability $1$. We now show that $W \upgraded W^-$. Recalling that $W$ is a BMS, we denote the corresponding permutation as $\pi$. We also denote by $\delta()$ a function taking as an argument a condition. The function $\delta$ equals $1$ if the condition is satisfied and $0$ otherwise. With these definitions at hand, we take
\ifCLASSOPTIONtwocolumn
  \begin{multline*}
    \Phi(y_0,y_1|y) \\
    = \frac{1}{2} \big[ W(y_1|0) \cdot \delta(y_0 = y) + W(y_1|1) \cdot \delta(y_0 = \pi(y)  \big] \; . \mbox{\qedhere{}}
  \end{multline*}
\else
\[
    \Phi(y_0,y_1|y)
    = \frac{1}{2} \big[ W(y_1|0) \cdot \delta(y_0 = y) + W(y_1|1) \cdot \delta(y_0 = \pi(y)  \big] \; . \mbox{\qedhere{}}
\]
\fi
\end{proof}
This is a good place to note that our algorithm is applicable to a slightly more general setting. Namely, the setting of compound polar codes as presented in \cite{Madhavifar2013}. The slight alterations needed are left to the reader.

The following lemma claims that both polar transformations preserve the upgradation relation. It is a restatement of \cite[Lemma 4.7]{Korada2009}.
\begin{lemma}
  Let $W:\calX \to \calY$ and $Q:\calX \to \calZ$ be two BMS channels such that $W \upgraded Q$. Then,
  \begin{equation}
    \label{eq:polarizationPreservesUpgradation}
    W^- \upgraded Q^- \qquad \mbox{and} \qquad W^+ \upgraded Q^+
  \end{equation}
\end{lemma}

For a BMS channel $W$ and $0 \leq i < n$, denote by $W_i^{(m)}$ the channel which is denoted ``$W_n^{(i+1)}$'' in \cite{Arikan2009}. By \cite[Proposition 13]{Arikan2009}, the channel $W_i^{(m)}$ is symmetric. The following lemma ties the two definitions of the $\upgraded$ relation.
\begin{lemma}
  Let $W:\calX \to \calY$ be a BMS channel. Let the indices $0 \leq i,j < n$ be given. Then, binary domination implies upgradation. That is,
  \begin{equation}
    \label{eq:binaryDominationImpliesUpgradation}
    i \dominates j  \quad \Longrightarrow \quad W_i^{(m)} \upgraded W_j^{(m)} \; .
  \end{equation}
\end{lemma}

\begin{proof}
  We prove the claim by induction on $m$. For $m=1$, the claim follows from either (\ref{eq:plusUpgradedWithRespectToMinus}), or the fact that a channel is upgraded with respect to itself, depending on the case. For $m > 1$, we have by induction that
  \[
  W_{\myfloorr{i/2}}^{(m-1)} \upgraded W_{\myfloorr{j/2}}^{(m-1)} \; .
  \]
  Now, if the least significant bits of $i$ and $j$ are the same we use (\ref{eq:polarizationPreservesUpgradation}), while if they differ we use (\ref{eq:plusUpgradedWithRespectToMinus}) and the transitivity of the ``$\upgraded$'' relation.
\end{proof}

We are now ready to prove our second main result.
\begin{theorem}
  \label{theo:polarCodesAreDominationContiguous}
  Let $A$ be the active rows set corresponding to a polar code. Then, $A$ is domination contiguous.
\end{theorem}
\begin{proof}
  We must first state exactly what we mean by a ``polar code''. Let the code dimension $k$ be specified. In \cite{Arikan2009}, $A$ equals the indices corresponding to the $k$ channels $W_i^{(m)}$ with smallest Bhattacharyya parameter, where $0 \leq i < n$. Other definitions are possible and will be discussed shortly. However, for now, let us use the above definition.

  Denote the Bhattacharyya parameter of a channel $W$ by $Z(W)$. As is well known, if $W$ and $Q$ are two BMS channels, then
  \begin{equation}
    \label{eq:upgradedImpliesBetterBhattacharyya}
    W \upgraded Q \quad \Longrightarrow \quad Z(W) \leq Z(Q) \; .
  \end{equation}
  For a proof of this fact, see \cite{Kailath1967}.

  We deduce from (\ref{eq:binaryDominationImpliesUpgradation}) and (\ref{eq:upgradedImpliesBetterBhattacharyya}) that if $i \dominates j$, then $Z(W_i^{(m)}) \leq Z(W_i^{(m)})$. Assume for a moment that the inequality is always strict when $i \upgraded j$ and $i \neq j$. Under this assumption, $j \in A$ must imply $i \in A$. This is a stronger claim then (\ref{eq:dominationContiguous}), which is the definition of $A$ being domination contiguous. Thus, under this assumption we are done.

  The previous assumption is in fact true for all relevant cases, but somewhat misleading: The set $A$ is constructed by algorithms calculating with finite precision. It could be the case that $i \neq j$, $i \upgraded j$, but $Z(W_i^{(m)})$ and $Z(W_i^{(m)})$ are approximated by the same number (a tie), or by two close numbers, but in the wrong order. Thus, it might conceptually be the case that $j$ is a member of $A$ while $i$ is not (in practice, we have never observed this to happen). These cases are easy to check and fix, simply by removing $j$ from $A$ and inserting $i$ instead. Note that each such operation enlarges the total Hamming weight of the vectors $\binary{t}$ corresponding to elements $t$ of $A$. Thus, such a swap operation will terminate in at most a finite number of steps. When the process terminates, we have by definition that if $j \in A$ and $i \upgraded j$, then $i \in A$. Thus, $A$ is dominations contiguous.

  Instead of taking the Bhattacharyya parameter as the figure of merit, we could have instead used the (more natural) channel misdecoding probability. That is, the probability of an incorrect maximum-likelihood estimation of the input to the channel given the channel output, assuming a uniform input distribution. Yet another figure of merit we could have taken is the channel capacity. The important point in the proof was that an upgraded channel has a figure of merit value that is no worse. This holds true for the other two options discussed in this paragraph. See \cite[Lemma 3]{Tal2011a} for details and references. We note that although these are natural ways of defining polar codes, they are in some cases sub-optimal. See for example \cite{Mondelli2014} (it is easily proved that our encoder is valid for the scheme in \cite{Mondelli2014} as well).
\end{proof}

\subsection{Application to Shortened Codes}
\label{sec:shortening}
We end this section by discussing two shortening procedures, \cite{Li2015} and \cite{Wang2014}, which are compatible with our systematic encoder. Recall that shortening a code at positions $\Gamma \subseteq \{0,1,\ldots,n-1\}$ means that only codewords $\bfx = (x_0,x_1,\ldots,x_{n-1})$ for which $\gamma \in \Gamma$ implies $x_\gamma = 0$ are part of the newly created code. Since the value of $\bfx$ at positions $\Gamma$ is known to be $0$, these codeword positions are not transmitted. Hence, a word of length $n - |\Gamma|$ is transmitted over the channel.

For our purposes, the polar shortening schemes \cite{Li2015} and \cite{Wang2014} are very similar. In \cite{Li2015}, the set $\Gamma$ is defined as
\[
\Gamma = \{\gamma : \gamma_0 \leq \gamma < n\} \; ,
\]
for a given $\gamma_0$. The encoding in \cite{Li2015} is of the non-bit-reversed type. In contrast, the set $\Gamma$ in \cite{Wang2014} is obtained from the set $\Gamma$ above by applying a bit-reversing operation. The encoding in \cite{Wang2014} is bit-reversed as well.

An important consequence of the above definition of $\Gamma$ is the following. In both settings, the shortening is accomplished by freezing the corresponding indices in the information vector. That is, $\gamma \in \Gamma$ implies that $\gamma \not\in A$. Also, in both settings, the ``channel'' corresponding to a position $\gamma \in \Gamma$ (which is not transmitted) is taken as the noiseless channel when constructing the polar code. The rational is that we know with certainty that the value of $x_\gamma$ is $0$. 

The applicability of our method to the above follows by two simple observations, which we now state without proof. Firstly, Lemma~\ref{lemm:plusUpgradedWithRespectToMinus} and its derivatives continue to hold in the setting in which the underlying channels may be of a different type. Specifically, note that at the lowest level, a plus or minus operation may involve a ``real'' channel and a ``noiseless'' channel. However, the natural analog of (\ref{eq:plusUpgradedWithRespectToMinus}) continues to hold. Namely, a plus operation is still upgraded with respect to a minus operation.

The second observation is that $i \upgraded j$ as well as $\reverse{i} \upgraded \reverse{j}$ imply that $i \geq j$. Thus, in this setting as well, domination contiguity continues to hold. Indeed, consider for concreteness the non-reversed case and  suppose to the contrary that (\ref{eq:dominationContiguous}) does not hold. Namely, we have found $h \upgraded i \upgraded j$ such that $h,j \in A$ but $i \not\in A$. By our first observation, $i \not\in A$ must be the result of $i$ being a shortened index, $i \in \Gamma$. But if $i$ is a shortened index, we have by our second observation that $h$ is a shortened index as well. Hence, $h \in \Gamma$ which implies that $h \not\in A$, contradiction.

\section{Flexible Hardware Encoders}
\label{sec:flex-enc}
The encoder discussed in the previous sections uses two instances---or two passes---of a non-systematic polar encoder to calculate the systematic codeword.
Therefore it is important to have a suitable non-systematic encoder that provides its output in natural or bit-reversed order.

A semi-parallel non-systematic polar encoder design with a throughput of $\mathcal{P}$ bit/cycle, where $\mathcal{P}$ corresponds to the level of parallelism, was presented in \cite{Yoo2015}.
However, it presents its output in pair bit-reversed order---the output is in bit-reversed order if a pair of consecutive bits is viewed as a single entity,---rendering it unsuitable for use with our systematic encoder. This also poses a problem for parallel and semi-parallel decoders, which expect their input either in natural or in bit-reversed order.

We start this section by presenting the architecture for a new non-systematic encoder that presents its output in natural order and has the same $\mathcal{P}$-bit/cycle throughput and $n/\mathcal{P}$-cycle latency as \cite{Yoo2015}. We show the impact of adding length flexibility support, and then utilize it as the core component of a flexible systematic encoder according to the algorithm discussed in this work.

\subsection{Non-Systematic Encoder Architecture}
\label{sec:felx-enc:arch}
Fig.~\ref{fig:enc-arch} shows the proposed architecture for a non-systematic encoder with $n = 16$ and $\mathcal{P} = 4$, where stage boundaries are indicated using dashed lines.
Each stage $S_i$, with index $i$, applies the basic polar transformation to two input bits, $\beta_{i-1}[j]$ and $\beta_{i-1}[j + 2^{i - 1}]$ that are $2^{i-1}$ bits apart in the polar code graph.
Since the input pairs to stages with indices $\in [1, \log \mathcal{P}]$ are available in the same $\mathcal{P}$-bit input and the same clock cycle, these stages are implemented using combinational logic only, as shown for $S_1$ and $S_2$ in the figure.
On the other hand, the two bits processed simultaneously by a stage with an index $i > \log \mathcal{P}$ are not available in the same clock cycle, necessitating the use of delay elements, denoted $D$ in Fig.~\ref{fig:enc-arch}.
Such a stage is implemented using $\mathcal{P}$ 1-bit processing elements operating in parallel, each of which has $2^{i - \log \mathcal{P} - 1}$ delay elements.
A processing element $l$ contains a multiplexer that alternates its output between $\beta_{i-1}[\mathcal{P}t + l - 2^{i-1}] \oplus \beta_{i-1}[\mathcal{P}t + l]$ and $\beta_{i-1}[\mathcal{P}t + l]$ every $2^{i - \log \mathcal{P} - 1}$ clock cycles, where $t$ is the current cycle index.

The resulting encoder has a throughput of $\mathcal{P}$ bit/s, a latency of $n/\mathcal{P}$ cycles, and a critical path that passes from $u[4t + 4]$ to $x[4t]$, similar to the encoder of \cite{Yoo2015}. The critical path can be shortened by inserting pipeline registers at stage boundaries, increasing latency in terms of cycles, but leaving throughput per cycle unaffected. In addition to the output order, the proposed architecture has another advantage over \cite{Yoo2015} in that it can be used to implement a fully serial encoder with $\mathcal{P} = 1$, whereas that of \cite{Yoo2015} can only scale down to $\mathcal{P} = 2$. We note that throughout this work, encoding latency is measured from first data-in to first data-out and all encoders start their operation as soon as the first $\mathcal{P}$ input bits are available.

\begin{figure}[t]
  \centering
  \begin{tikzpicture}[>=stealth]

  \usetikzlibrary{positioning,calc,mux}

  \definecolor{varnode_fill}{RGB}{0,0,0}
  \definecolor{chknode_fill}{RGB}{255,255,255}

  \tikzset{
    chknode/.style={draw,fill=chknode_fill,circle,minimum size=0.3cm, inner sep=0},
    varnode/.style={draw,fill=varnode_fill,circle,minimum size=0.1cm, inner sep=0},
    sep/.style={rectangle,minimum width=0cm, inner sep=0},
    empty/.style={rectangle, inner sep=0},
    bit/.style={circle, inner sep = 0},
    delay/.style={draw,rectangle,inner sep =1mm},
  }
  \tikzset{every mux2 node/.style={draw,minimum width=0.2cm,minimum height=0.6cm,inner sep=1mm,outer sep=0pt}}

  \node[bit] at (0.6, 3.2)  {\scriptsize$u[4t]$};
  \node[bit] at (0.6, 2.2)  {\scriptsize$u[4t\!+\!1]$};
  \node[bit] at (0.6, 1.2)  {\scriptsize$u[4t\!+\!2]$};
  \node[bit] at (0.6, 0.2)  {\scriptsize$u[4t\!+\!3]$};

  \node[coordinate] at (0.7, 3) (u0) {};
  \node[coordinate] at (0.7, 2) (u1) {};
  \node[coordinate] at (0.7, 1) (u2) {};
  \node[coordinate] at (0.7, 0) (u3) {};


  \node[chknode] at (1.2, 3) (n1_0) {$+$};
  \node[varnode] at (1.2, 2) (n1_1) {};
  \node[chknode] at (1.2, 1) (n1_2) {$+$};
  \node[varnode] at (1.2, 0) (n1_3) {};

  \draw[->] (u0) -- (n1_0);
  \draw[-] (u1) -- (n1_1);
  \draw[->] (u2) -- (n1_2);
  \draw[-] (u3) -- (n1_3);

  \draw[->] (n1_1) -- (n1_0);
  \draw[->] (n1_3) -- (n1_2);


  \node[chknode] at (2.3, 3) (n2_0) {$+$};
  \node[chknode] at (1.8, 2) (n2_1) {$+$};
  \node[varnode] at (2.3, 1) (n2_2) {};
  \node[varnode] at (1.8, 0) (n2_3) {};

  \draw[->] (n1_0) -- (n2_0);
  \draw[->] (n1_1) -- (n2_1);
  \draw[-] (n1_2) -- (n2_2);
  \draw[-] (n1_3) -- (n2_3);

  \draw[->] (n2_2) -- (n2_0);
  \draw[->] (n2_3) -- (n2_1);


  \foreach \i in {0, 1, 2, 3} {
    \node[varnode] at (2.8, 3 - \i) (n3_\i) {};
    \node[delay] at (3.3, 3 - \i) (d3_\i) {$D$};
    \node[varnode] at (3.8, 3 - \i) (ne3_\i) {};
    \node[chknode] at (4.3, 3 - \i) (nc3_\i) {$+$};
    \node[mux2] at (4.8, 3 - \i + 0.15) (m3_\i) {};

    \node at ($(m3_\i.in0) + (0.06, 0)$) {\tiny0};
    \node at ($(m3_\i.in1) + (0.06, 0)$) {\tiny1};
  
    \node[coordinate] at (5, 3 - \i) (o3_\i) {};

    \draw[->] (n3_\i) -- (d3_\i);
    \draw[-] (d3_\i) -- (ne3_\i);
    \draw[->] (ne3_\i) -- (nc3_\i);
    \draw[->] (nc3_\i) -- (m3_\i.in1);

    \draw[->] (n3_\i) -- ++(0, -0.4) -| (nc3_\i);
    \draw[->] (ne3_\i) |- (m3_\i.in0);

    \draw[-] (m3_\i.out) -| (o3_\i);

    \draw[-] (n2_\i) -- (n3_\i);
  }


  \foreach \i in {0, 1, 2, 3} {
    \node[varnode] at (5.3, 3 - \i) (n4_\i) {};
    \node[delay] at (5.8, 3 - \i) (d4_\i) {$D$};
    \node[delay] at (6.3, 3 - \i) (dd4_\i) {$D$};
    \node[varnode] at (6.8, 3 - \i) (ne4_\i) {};
    \node[chknode] at (7.3, 3 - \i) (nc4_\i) {$+$};
    \node[mux2] at (7.8, 3 - \i + 0.15) (m4_\i) {};

    \node at ($(m4_\i.in0) + (0.06, 0)$) {\tiny0};
    \node at ($(m4_\i.in1) + (0.06, 0)$) {\tiny1};

    \node[coordinate] at (m4_\i.out) (o4_\i) {};
  
    \draw[->] (n4_\i) -- (d4_\i);
    \draw[-] (dd4_\i) -- (ne4_\i);
    \draw[->] (ne4_\i) -- (nc4_\i);
    \draw[->] (nc4_\i) -- (m4_\i.in1);

    \draw[->] (n4_\i) -- ++(0, -0.4) -| (nc4_\i);
    \draw[->] (ne4_\i) |- (m4_\i.in0);

    \draw[-] (o3_\i) -- (n4_\i);
  }

  \foreach \i in {0, 1, 2, 3} {
    \node[coordinate] at (8.5, 3 - \i + 0.15) (x_\i) {}; 
    \draw[->] (o4_\i) -- (x_\i);
  }
  \node[bit] at (8.5, 3 - 0.15) {\scriptsize$x[4t]$};
  \node[bit] at (8.5, 2 - 0.15) {\scriptsize$x[4t\!+\!1]$};
  \node[bit] at (8.5, 1 - 0.15) {\scriptsize$x[4t\!+\!2]$};
  \node[bit] at (8.5, 0 - 0.15) {\scriptsize$x[4t\!+\!3]$};


  \node[coordinate] at (1.5, 3.8) (l1_t) {};
  \node[coordinate] at (1.5, -0.8) (l1_b) {};
  \draw[-,dashed] (l1_t) -- (l1_b);
  \node at (1, 3.8) {$S_1$};

  \node[coordinate] at (2.6, 3.8) (l2_t) {};
  \node[coordinate] at (2.6, -0.8) (l2_b) {};
  \draw[-,dashed] (l2_t) -- (l2_b);
  \node at (2, 3.8) {$S_2$};

  \node[coordinate] at (5.1, 3.8) (l3_t) {};
  \node[coordinate] at (5.1, -0.8) (l3_b) {};
  \draw[-,dashed] (l3_t) -- (l3_b);
  \node at (3.8, 3.8) {$S_3$};

  \node at (6.8, 3.8) {$S_4$};
\end{tikzpicture}
  \caption{Architecture of the proposed semi-parallel non-systematic polar encoder with $n = 16$ and $\mathcal{P} = 4$.}
  \label{fig:enc-arch}
\end{figure}
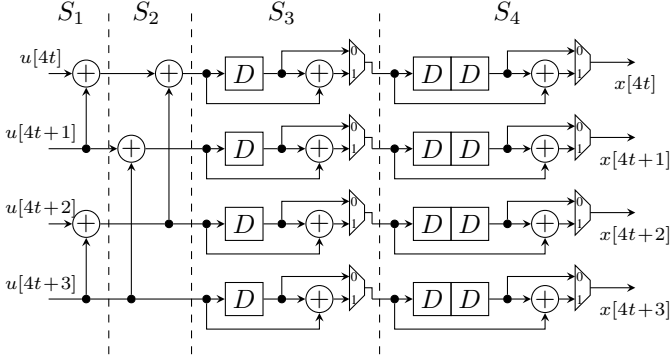

\subsection{Flexible Non-Systematic Encoder}
Since the input $\bfu$ is assumed to contain `0's in the frozen bit locations, the proposed encoder is rate flexible as the input preprocessor can change the location and number of frozen bits without affecting the encoder architecture.

Adapting this architecture to encode any polar code of length $n \leq n_{\text{max}}$ requires extracting data from different stage outputs---indicated using the dashed lines in Fig.~\ref{fig:enc-arch}---in the encoder. The output for a code of length $n$ can be extracted from location $S_{\log n}$ using $\mathcal{P}$ instances of a $\log n_{\text{max}} \times 1$ multiplexer.

The width of the multiplexer can be reduced to $\log n_{\text{max}} - \log \mathcal{P}$ without affecting decoding latency by exploiting the combinational nature of stages $\in [S_1, S_{\log{\mathcal{P}}}]$ and setting inputs with indices $i > n - 1$ to `0'. The modified encoder architecture is illustrated in Fig.~\ref{fig:flex-ns-enc}, where the block labeled `encoder' is the non-systematic encoder shown in Fig.~\ref{fig:enc-arch}. The AND gates are used to mask inputs when $n < \mathcal{P}$. The first AND gate, $\&_0$, will always have it second input set to `1' in this case and will be optimized away by the synthesis tool. It is shown in the figure because it will be used to implement code shortening as described in Section~\ref{sec:enc-short}.

\begin{figure}[t]
  \centering
  \usetikzlibrary{circuits.logic.US}
\begin{tikzpicture}[>=stealth, circuit logic US]

  \usetikzlibrary{positioning,calc,mux}

  \tikzset{
    bit/.style={circle, inner sep = 0},
  }


  \node[and gate, inputs={normal,normal},xscale=0.85] at (-0.2, 0) (and0) {\scriptsize $\&_0$};
  \node[and gate, inputs={normal,normal},xscale=0.85] at (-0.2, -1) (and1) {\scriptsize $\&_1$};
  \node[] at (-0.2, -2) () {$\vdots$};
  \node[and gate, inputs={normal,normal},xscale=0.85] at (-0.2, -3) (andp) {};


  \node[coordinate] at ($(and0.input 1) - (0.3, 0)$) (u0c) {};
  \node[bit,anchor=east] at ($(u0c) + (0, 0.05)$) {\scriptsize$u[\mathcal{P}t]$};

  \node[coordinate] at ($(and1.input 1) - (0.3, 0)$) (u1c) {};
  \node[bit,anchor=east] at ($(u1c) + (0, 0.05)$) {\scriptsize$u[\mathcal{P}t\!+\!1]$};

  \node[coordinate] (upc) at ($(andp.input 1) - (0.3, 0)$) {};
  \node[bit,anchor=east] at ($(upc) + (0, 0.05)$) {\scriptsize$u[\mathcal{P}t\!+\!\mathcal{P}\!-\!1]$};

  \node[coordinate] (um0c) at ($(and0.input 2) - (0.3, 0)$) {};
  \node[bit,anchor=east] at ($(um0c) - (0, 0.05)$) {\scriptsize$n \geq 0$};

  \node[coordinate] (um1c) at ($(and1.input 2) - (0.3, 0)$) {};
  \node[bit,anchor=east] at ($(um1c) - (0, 0.05)$) {\scriptsize$n \geq 1$};

  \node[coordinate] (umpc) at ($(andp.input 2) - (0.3, 0)$) {};
  \node[bit,anchor=east] at ($(umpc) - (0, 0.05)$) {\scriptsize$n \geq \lfloor\frac{\mathcal{P}\!-\!1}{2}\rfloor$};

  \draw[->] (u0c) -- (and0.input 1);
  \draw[->] (um0c) -- (and0.input 2);

  \draw[->] (u1c) -- (and1.input 1);
  \draw[->] (um1c) -- (and1.input 2);

  \draw[->] (upc) -- (andp.input 1);
  \draw[->] (umpc) -- (andp.input 2);

  \node[rectangle,draw,minimum height=4cm, minimum width=2] at (1, -1.5) (enc) {\scriptsize encoder};

  \draw let \p1=(enc.west),\p2=(and0.output) in node[coordinate] (c0) at (\x1,\y2) {};
  \draw[->] (and0.output) -- (c0);

  \draw let \p1=(enc.west),\p2=(and1.output) in node[coordinate] (c1) at (\x1,\y2) {};
  \draw[->] (and1.output) -- (c1);

  \draw let \p1=(enc.west),\p2=(andp.output) in node[coordinate] (cp) at (\x1,\y2) {};
  \draw[->] (andp.output) -- (cp);

  \node[mux4,draw,minimum height=4cm, minimum width=5mm] (m0) at (3, -1.5) {};

  \draw let \p1=(enc.east),\p2=(m0.in0) in node[coordinate] (o0) at (\x1, \y2) {};
  \draw[->] (o0) -- node[coordinate] (b0) {} (m0.in0);
  \node[bit] at ($(b0) + (0, 0.2)$) {\scriptsize$\beta_{\log \mathcal{P}}$};

  \draw let \p1=(enc.east),\p2=(m0.in1) in node[coordinate] (o1) at (\x1, \y2) {};
  \draw[->] (o1) -- node[coordinate] (b1) {} (m0.in1);
  \node[bit] at ($(b1) + (0, 0.2)$) {\scriptsize$\beta_{\log \mathcal{P} + 1}$};

  \draw let \p1=(enc.east),\p2=(m0.in2) in node[coordinate] (o2) at (\x1, \y2) {};
  \node[bit] at ($(o2)!0.5!(m0.in2)$) {$\vdots$};

  \draw let \p1=(enc.east),\p2=(m0.in3) in node[coordinate] (o3) at (\x1, \y2) {};
  \draw[->] (o3) -- node[coordinate] (b3) {} (m0.in3);
  \node[bit] at ($(b3) + (0, 0.2)$) {\scriptsize$\beta_{\log n_{\text{max}}}$};

  \draw let \p1=(m0.out) in node[coordinate] (out) at (3.5, \y1) {};
  \draw[->] (m0.out) -- (out);

  \node[bit,anchor=west] at (out) {\scriptsize$x[\mathcal{P}t, \cdots, \mathcal{P}t + \mathcal{P} - 1]$};

  \node[rectangle,inner sep=1mm] at ($(m0.sel) - (0, 1)$) (sel) {\scriptsize$\log \lceil \frac{n}{\mathcal{P}} \rceil$};
  \draw[->] (sel) -- ++ (0, 0.78);

  \node[bit] at ($(m0.in0) + (0.1, 0)$) {\scriptsize 0};
  \node[bit] at ($(m0.in1) + (0.1, 0)$) {\scriptsize 1};

\end{tikzpicture}
  \caption{Flexible encoder with maximum code length $n_{\text{max}}$ and parallelism $\mathcal{P}$.}
  \label{fig:flex-ns-enc}
\end{figure}
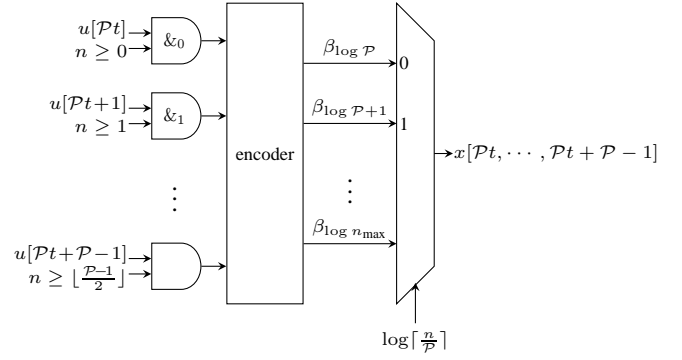

\subsection{Non-Systematic Encoder Implementation}
\label{sec:ns-impl}
Both the rate-flexible and rate and length-flexible versions of the proposed non-systematic encoder were implemented on the Altera Stratix IV EP4SGX530KH40C2 FPGA. We also implemented the encoder of \cite{Yoo2015} on the same FPGA for comparison even though its output order is not suitable for implementing the proposed systematic encoder.
All decoders have a latency of 512 cycles and coded throughput of 32 bits/cycle for $n_{\text{max}} = 16384$ and $\mathcal{P} = 32$.
Table~\ref{tab:ns-impl} presents these implementation results, where the proposed rate-flexible and the rate and length-flexible encoders are denoted $R$-flexible and $Rn$-flexible, respectively.
From the table it can be observed that including length flexibility increases the logic requirements of the design by 27\% due to the extra routing required. It also decreases the maximum achievable frequency, and in turn throughput, by 14\%.
The latency of the decoders is bounded by $n_{\text{max}}/\mathcal{P}$ and increasing $\mathcal{P}$ to 64 reduced it to 0.74 and 0.82 $\mu$s, increasing the throughput to 22 and 20 Gbps, for the $R$- and $Rn$-flexible encoders, respectively. The maximum achievable frequency decreased to 344 and 313 MHz for the two encoders.

\begin{table}[t]
  \centering
  \caption{Implementation of the proposed $R$-flexible and $Rn$-flexible non-systematic encoders compared with that of \cite{Yoo2015} for $n_{\max} = 16384$ and $\mathcal{P} = 32$ on the Altera Stratix IV EP4SGX530KH40C2.}
  \begin{tabular}{r c c c c c c}
    \toprule
    Decoder & LUTs & FF & RAM    & $f$   & Lat.     & T/P\\
            &      &    & (bits) & (MHz) & ($\mu$s) & (Gbps)\\
    \midrule
    \cite{Yoo2015} & 769 & 1,392 & 12,160 & 354 & 1.4 & 11.3\\
    $R$-flexible & 649 & 1,240 & 12,160 & 394 & 1.3 & 12.6\\
    $Rn$-flexible & 838 & 1,293 & 12,160 & 360 & 1.4 & 11.5\\
    \bottomrule
  \end{tabular}
  \label{tab:ns-impl}
\end{table}

The results were obtained using Altera Quartus II 15.0 and verified using both RTL and gate-level simulation with randomized testbenches.

\subsection{Systematic Encoder Architecture}
\label{sec:sys-arch}
With a non-systematic encoder providing its output in a suitable order, we now present the architecture and implementation results for the proposed systematic encoder. As proved in Sections \ref{sec:systematicEncodingDefs} and \ref{sec:dominationContiguiryImpliesInvolution}, the proposed systematic encoder can present its output with parity bits in bit-reversed or natural order locations---even with the non-systematic encoder providing its output in natural order---by changing the location of the frozen bits. We therefore use an $n_{\text{max}}/\mathcal{P} \times \mathcal{P}$-bit memory to store the frozen-bit mask, enabling the encoder to support both parity-bit locations, in addition to rate flexibility.

As mentioned in Section~\ref{subsec:systematicEncoderDescription}, the systematic-encoding process performs two non-systematic encoding passes on the data. These passes can be implemented using two instances of the proposed non-systematic encoder. The output of the first is stored in registers and then masked according the content of the mask memory before being passed to another level of pipeline registers to limit the critical path length. The output of the registers is then passed to the second non-systematic encoder instance, whose output forms the systematic codeword. Such an architecture has the same $\mathcal{P}$-bit/cycle throughput of the component non-systematic encoder with a latency $\mathcal{L} = 2\mathcal{L}_{\text{NS}} + 2$ cycles, where $\mathcal{L}_{\text{NS}}$ is the latency of the non-systematic encoder.

Alternatively, to save implementation resources at the cost of halving the throughput, one instance of the component encoder can be used for both passes. The output of the non-systematic encoder is stored in registers after the first pass and is routed back to the input of the encoder. The systematic codeword becomes available after the second pass.

The systematic encoder of \cite{Arikan2011} can be used in a configuration similar to the proposed high-throughput one. However, it requires multiplication by matrices that change when the frozen bits are changed. Therefore, its implementation requires a configurable parallel matrix multiplier that is significantly more complex than the component non-systematic encoder used in this work. When the encoder of \cite{Arikan2011} is implemented to be rate-flexible and low-complexity, it has a latency of at least $n$ clock cycles; compared to the $2n/\mathcal{P} + 2$ cycle latency of the proposed architecture.

\subsection{Systematic Encoder Implementation}
\label{sec:sys-impl}
Implementation results of the throughput oriented $R$-flexible and $Rn$-flexible encoders are presented in Table~\ref{tab:sys-impl} both with and without pipeline registers in between the two non-systematic encoder instances. In the pipelined version two levels were used: one before and one after the masking operation, since memory access incurred a comparatively long delay. The results show that the pipelined version performs significantly faster than the non-pipelined version, where the clock frequency was increased by 80 MHz for both pipelined encoders and was limited by clock and asynchronous reset distribution. The pipelining yielded throughput values of 9 and 8.4 Gbps for the $R$-flexible and $Rn$-flexible encoders, respectively. The reported amount of RAM included the mask memory, in addition to operations that were converted automatically by the synthesis and mapping tools.

\begin{table}[t]
  \centering
  \caption{Implementation of the proposed $R$-flexible and $Rn$-flexible systematic encoders for $n_{\max} = 16384$ and $\mathcal{P} = 32$ on the Altera Stratix IV EP4SGX530KH40C2.}
  \begin{tabular}{r c c c c c c}
    \toprule
    Decoder & LUTs & FF & RAM    & $f$   & Lat.     & T/P\\
            &      &    & (bits) & (MHz) & ($\mu$s) & (Gbps)\\
    \midrule
    \textbf{Non-Pipelined} & & & & \\
    $R$-flexible  & 1,442 & 2,320 & 36,924 & 206 & 5.0 & 6.6\\
    $Rn$-flexible & 1,782 & 2,381 & 36,924 & 180 & 5.7 & 5.7\vspace{2pt}\\
    \textbf{Pipelined} & & & & \\
    $R$-flexible  & 1,397 & 2,639 & 36,924 & 282 & 3.6 & 9.0\\
    $Rn$-flexible & 1,606 & 2,742 & 36,924 & 264 & 3.9 & 8.4\\
    \bottomrule
  \end{tabular}
  \label{tab:sys-impl}
\end{table}

As in the case of the non-systematic encoder, the throughput is proportional to $\mathcal{P}$ and the latency to $n / \mathcal{P}$. Table~\ref{tab:sys-impl-var} explores the effect of different $n_{\text{max}}$ and $\mathcal{P}$ values on the pipelined $Rn$-flexible encoder. Throughput in excess of 10 Gbps is achievable by the encoder when $\mathcal{P} > 32$. When $n < n_{\text{max}}$, throughput remains unchanged and latency decreases to $n/n_{\text{max}}$ of its original value. For example, when the encoder with $n_{\text{max}} = 16384$ and $\mathcal{P} = 64$ encodes a code with $n = 2048$, throughput remains 15 Gbps and latency decreases to 281 ns.

\begin{table}[t]
  \centering
  \caption{Implementation of the proposed $Rn$-flexible systematic encoder for different $n_{\max}$ and $\mathcal{P}$ values on the Altera Stratix IV EP4SGX530KH40C2.}
  \begin{tabular}{c c c c c c c c}
    \toprule
     $n_{\text{max}}$ & $\mathcal{P}$ & LUTs & FF & RAM    & $f$   & Lat.     & T/P\\
                   &               &      &    & (bits) & (MHz) & ($\mu$s) & (Gbps)\\
    \midrule
    16,384 & 32  & 1,606 & 2,742  & 36,924 & 264 & 3.9 & 8.4\\
    16,384 & 64  & 2,872 & 5,287  & 16,384 & 235 & 2.2 & 15.0\\
    16,384 & 128 & 4,404 & 8,304  & 16,384 & 272 & 0.9 & 34.8\\
    32,768 & 32  & 1,971 & 2,997  & 85,948 & 258 & 7.9 & 8.2\\
    32,768 & 64  & 3,390 & 5,601  & 64,200 & 265 & 3.9 & 16.9\\
    32,768 & 128 & 5,550 & 10,024 & 37,304 & 234 & 2.2 & 29.9\\
    \bottomrule
  \end{tabular}
  \label{tab:sys-impl-var}
\end{table}

\subsection{On Code Shortening}
\label{sec:enc-short}
As discussed in Subsection~\ref{sec:shortening}, the works in \cite{Li2015} and \cite{Wang2014} describe shortening schemes for polar codes in which the last $n - n_s$ information bits in a polar code of length $n$ are replaced with `0's. Those bits are discarded from the systematic codeword before transmission.
The result is that $n_s$ bits containing $k_s$ information bits are transmitted; where $k_s = k - (n - n_s)$.
These schemes are suitable for use with the proposed systematic encoder, yielding a system that can encode normal and shortened polar codes of any length $n \in [2, n_{\text{max}}]$ without any other constraints on the code length or rate.

To adapt our proposed systematic encoder and enable shortening, the second input, $\text{en}_i$, to the AND gates $\&_i$ becomes
\[
\text{en}_i = \begin{cases}
  1 & \text{when } n \geq \lfloor (\mathcal{P}t + i)/2 \rfloor,\\
  0 & \text{otherwise.}
\end{cases}
\]

Adding code shortening ability has a minor effect on the resource utilization of the $Rn$-flexible encoder as can be observed in Table~\ref{tab:sys-short}.

\begin{table}[t]
  \centering
  \caption{Implementation of the proposed $Rn$-flexible systematic encoder with shortening on the Altera Stratix IV EP4SGX530KH40C2.}
  \begin{tabular}{c c c c c c c c}
    \toprule
    $n_{\text{max}}$ & $\mathcal{P}$ & LUTs & FF & RAM    & $f$   & Lat.     & T/P\\
                  &               &      &    & (bits) & (MHz) & ($\mu$s) & (Gbps)\\
    \midrule
    16,384 & 128 & 4,518 & 8,667  & 16,384 & 272 & 0.9 & 34.8\\
    \bottomrule
  \end{tabular}
  \label{tab:sys-short}
\end{table}

\section{Flexible Software Encoders}
\label{sec:flex-sw-enc}
In this section, we present a software implementation of our systematic encoder using single-instruction multiple-data (SIMD) operations. We use both AVX (256-bit) and SSE (128-bit) SIMD extensions, in addition to the built-in types \texttt{uint8\_t}, \texttt{uint16\_t}, \texttt{uint32\_t}, and \texttt{uint64\_t} to operate on multiple bits simultaneously. The width of the selected type determines the encoder parallelism parameter $\mathcal{P}$, e.g. $\mathcal{P} = 8$ for \texttt{uint8\_t}.

The component non-systematic encoder progresses from stage $S_1$ to $S_{\log n}$ and presents its output in natural order. The input to $S_1$ is a packed vector where bits corresponding to frozen locations are set to `0' and information bits are stored in the other locations.
The bit with index $t$ at the output of a stage $S_i$ is calculated according to:
\begin{equation}
\beta_i[t] = \begin{cases}
  \beta_{i - 1}[t - 2^{i - 1}] \oplus \beta_{i - 1}[t] & \text{when } \lfloor t/2^{i - 1}\rfloor \text{ is even,}\\
  \beta_{i - 1}[t] & \text{otherwise.}
\end{cases}
\label{eq:enc-combine}
\end{equation}
This operation is applied directly to $\mathcal{P}$ bits simultaneously in stage $S_i$ if $2^{i - 1} \geq \mathcal{P}$. However, since we can only read and write data in groups of $\mathcal{P}$ bits whose addresses are aligned to $\mathcal{P}$-bit boundaries, operations in stages $S_i$ with $2^{i - 1} < \mathcal{P}$ are performed using a mask-shift-XOR procedure. A $\mathcal{P}$-bit mask $m_i$ is generated for each stage $\in [S_1, S_{\log n}]$ so that:
\[
m_i[t] = \begin{cases}
  0 & \text{when } \lfloor t/2^{i - 1}\rfloor \text{ is even,}\\
  1 & \text{otherwise.}
\end{cases}
\]
The output for these stages is calculated using:
\ifCLASSOPTIONtwocolumn
\begin{align}
\beta_i[t:t + \mathcal{P} - 1] = &\beta_{i-1}[t:t + \mathcal{P} - 1] \oplus \nonumber\\
& ((\beta_{i-1}[t:t + \mathcal{P} - 1]\;\&\;m_i) >> 2^{i - 1}). \nonumber
\end{align}
\else
\[
\beta_i[t:t + \mathcal{P} - 1] = \beta_{i-1}[t:t + \mathcal{P} - 1] \oplus
 ((\beta_{i-1}[t:t + \mathcal{P} - 1]\;\&\;m_i) >> 2^{i - 1}). 
\]
\fi
The index $t$ starts at 0 and is incremented by $\mathcal{P}$ with a final value of $N - \mathcal{P}$. The group of $\mathcal{P}$ bits with indices $\in [t, t + \mathcal{P} - 1]$ is denoted $t:t + \mathcal{P} - 1$. The symbol $\&$ is the bit-wise binary AND operation, and $>>$ is the logical bit right shift operator.

Since SSE operations lack bit shift operations, but include byte shifts, operations for stages ${S_1, S_2, S_3}$ are performed using the \texttt{uint64\_t} native type in the proposed software encoder. AVX version 1 does not provide any shift operations, and version 2 can only perform byte-shifts within 128-bit lanes. Therefore, we use SSE instructions until stage $S_9$, where the encoder switches to using AVX operations.
The masking operation between the two non-systematic encoding passes is applied using $\mathcal{P}$-bit operations and masks.

The resulting software systematic encoder operates on data in-place and requires $n$ bits of additional memory to store the frozen-bit mask, and another $\mathcal{P} \log \mathcal{P}$ bits to store the stage masks. The latency and coded throughput values for the proposed software systematic encoder running on a 3.4 GHz Intel Core i7-2600 are shown in Table~\ref{tab:sw-sys-impl} for $n = 32,768$. $\mathcal{P}$ was varied between 8 and 256. It can be seen that the latency decreases linearly with increasing $\mathcal{P}$ until $\mathcal{P} = 128$. The latency only decreases by 20\% between the SSE (128-bit) and AVX (256-bit) encoders for two reasons: the use of SSE for stages to up $S_9$ in the AVX encoder, and the overhead of loops and conditionals in the encoder. As a result, an encoder specialized for a given $n$ value is expected to be faster.

The speed results indicate that even an embedded 8-bit micro-processor running 1000 times slower would still be capable of transmissions at 500 kbps, eliminating the need for a dedicated hardware encoder for many applications such as remote sensors and some internet of things devices.

\begin{table}[t]
  \centering
  \caption{Latency and coded throughput of a software systematic encoder with $n = 32,768$ and different $\mathcal{P}$ values running on an Intel Core i7-2600.}
  \begin{tabular}{r c c c c c c}
    \toprule
    $\mathcal{P}$ & 8 & 16 & 32 & 64 & 128 & 256\\
    \midrule
    Latency ($\mu$s) & 64.1 & 30.1 & 14.3 & 7.7 & 4.1 & 3.3\\
    T/P (Gbps) & 0.5 & 1.1 & 2.3 & 4.2 & 8.0 & 10.0\\
    \bottomrule
  \end{tabular}
  \label{tab:sw-sys-impl}
\end{table}

\section{Flexible Hardware Decoders}
\label{sec:hw-dec}
We complete the flexible hardware polar coding system in this section by presenting flexible, systematic hardware decoder, which can decode channel messages based on the codewords generated by the encoder presented in Section~\ref{sec:flex-enc}. As discussed in \cite{Sarkis2014}, it is important that the parity bits be in bit-reversed locations to reduce routing complexity and simplify memory accesses.

The original Fast-SSC decoder was capable of decoding all polar codes of a given length: it resembled a processor where the polar code is loaded as a set of instructions \cite{Sarkis2014}.
In this section, we review the Fast-SSC algorithm, describe the architectural modifications necessary to decode any polar code up to a maximum length $n_{\text{max}}$ and analyze the resulting implementation.

\subsection{The Fast-SSC Decoding Algorithm}
\label{sec:bg:fast-ssc}
Our proposed flexible decoders utilize the Fast-SSC presented in \cite{Sarkis2014}. The polar code is viewed as a tree that corresponds to the recursive nature of polar code construction: a polar code of length $n$ is the concatenation of two polar codes of length $n/2$.
In the successive cancellation decoding, the tree is traversed depth first starting from stage $S_{\log n}$ until leaf node in stage $S_0$, corresponding to a constituent code of length $n = 1$ is reached. At that point, the output of is `0' if the leaf node corresponds to a frozen bit. Otherwise, it is calculated from the input log likelihood ratio (LLR) based on threshold detection.
The SSC decoding algorithm, \cite{Alamdar-Yazdi2011}, directly decodes constituent codes of any length that are of rate 0 or rate 1 without traversing their sub-trees. The Fast-SSC algorithm directly decodes single parity check (SPC) and repetition codes, in addition to rate-0 and rate-1 codes, of any length. Fig.~\ref{fig:fast-ssc} shows the Fast-SSC tree of an (8, 4) polar code, where the flow of messages is indicated by the arrows.

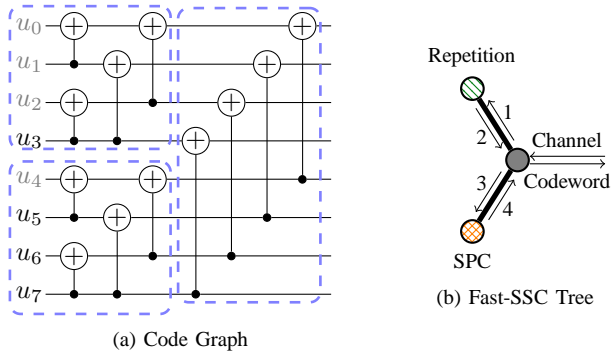
\begin{figure}[t]
  \centering
  \subfloat[Code Graph]{\label{fig:sc-graph}\newcommand{\ubit}[1]{$u_{#1}$}
\newcommand{\fbit}[1]{\color{gray}$u_{#1}$}
\begin{tikzpicture}[baseline=(base.center)]

\usetikzlibrary{shapes,positioning,arrows,decorations.markings,fit}

\definecolor{varnode_fill}{RGB}{0,0,0}
\definecolor{chknode_fill}{RGB}{255,255,255}

\tikzset{
  chknode/.style={draw,fill=chknode_fill,circle,minimum size=0.3cm, inner sep=0},
  varnode/.style={draw,fill=varnode_fill,circle,minimum size=0.1cm, inner sep=0},
  sep/.style={rectangle,minimum width=0mm, inner sep=0},
  bit/.style={circle, inner sep = 0}
}

\tikzset{blue dotted/.style={draw=blue!50!white, line width=1pt,
    dash pattern=on 4pt off 4pt,
    inner sep=0.5mm, rectangle, rounded corners}};

\tikzset{blue dotted tight/.style={draw=blue!50!white, line width=1pt,
    dash pattern=on 4pt off 4pt,
    inner sep=0mm, rectangle, rounded corners}};

\matrix[row sep=1mm, column sep=1mm] {
	\node[bit] (n0s0) {\fbit{0}}; & \node[sep] (s0) {}; & \node[chknode] (n0s1) {$+$}; & \node[sep] (s10) {}; &&	\node[chknode] (n0s2) {$+$}; &	\node[sep] (s20) {}; &&&& \node[chknode] (n0s3) {$+$}; &    \node[sep] (s30) {}; & \node[circle] (n0s4) {}; \\ 
	\node[bit] (n1s0) {\fbit{1}}; & \node[sep] (s1) {}; & \node[varnode] (n1s1) {};	   & \node[sep] (s11) {}; &	 \node[chknode] (n1s2) {$+$}; && \node[sep] (s21) {}; &&&  \node[chknode] (n1s3) {$+$}; &&   \node[sep] (s31) {}; & \node[circle] (n1s4) {}; \\ 
	\node[bit] (n2s0) {\fbit{2}}; & \node[sep] (s2) {}; & \node[chknode] (n2s1) {$+$}; & \node[sep] (s12) {}; &&	 \node[varnode] (n2s2) {};    &	 \node[sep] (s22) {}; &&   \node[chknode] (n2s3) {$+$}; &&&  \node[sep] (s32) {}; & \node[circle] (n2s4) {}; \\ 
	\node[bit] (n3s0) {\ubit{3}}; & \node[sep] (s3) {}; & \node[varnode] (n3s1) {};	   & \node[sep] (s13) {}; &	 \node[varnode] (n3s2) {};    && \node[sep] (s23) {}; &	   \node[chknode] (n3s3) {$+$}; &&&& \node[sep] (s33) {}; & \node[circle] (n3s4) {}; \\ 
	\node[bit] (n4s0) {\fbit{4}}; & \node[sep] (s4) {}; & \node[chknode] (n4s1) {$+$}; & \node[sep] (s14) {}; &&	 \node[chknode] (n4s2) {$+$}; &	 \node[sep] (s24) {}; &&&& \node[varnode] (n4s3) {};	&    \node[sep] (s34) {}; & \node[circle] (n4s4) {}; \\ 
	\node[bit] (n5s0) {\ubit{5}}; & \node[sep] (s5) {}; & \node[varnode] (n5s1) {};	   & \node[sep] (s15) {}; &	 \node[chknode] (n5s2) {$+$}; && \node[sep] (s25) {}; &&&  \node[varnode] (n5s3) {};	&&   \node[sep] (s35) {}; & \node[circle] (n5s4) {}; \\ 
	\node[bit] (n6s0) {\ubit{6}}; & \node[sep] (s6) {}; & \node[chknode] (n6s1) {$+$}; & \node[sep] (s16) {}; &&	 \node[varnode] (n6s2) {};    &	 \node[sep] (s26) {}; &&   \node[varnode] (n6s3) {};	&&&  \node[sep] (s36) {}; & \node[circle] (n6s4) {}; \\ 
	\node[bit] (n7s0) {\ubit{7}}; & \node[sep] (s7) {}; & \node[varnode] (n7s1) {};	   & \node[sep] (s17) {}; &	 \node[varnode] (n7s2) {};    && \node[sep] (s27) {}; &	   \node[varnode] (n7s3) {};	&&&& \node[sep] (s37) {}; & \node[circle] (n7s4) {}; \\ 
};
\path[-] (n0s0) edge (n0s1) (n0s1) edge (n0s2) (n0s2) edge (n0s3) (n0s3) edge (n0s4);
\path[-] (n1s0) edge (n1s1) (n1s1) edge (n1s2) (n1s2) edge (n1s3) (n1s3) edge (n1s4);
\path[-] (n2s0) edge (n2s1) (n2s1) edge (n2s2) (n2s2) edge (n2s3) (n2s3) edge (n2s4);
\path[-] (n3s0) edge (n3s1) (n3s1) edge (n3s2) (n3s2) edge (n3s3) (n3s3) edge (n3s4);
\path[-] (n4s0) edge (n4s1) (n4s1) edge (n4s2) (n4s2) edge (n4s3) (n4s3) edge (n4s4);
\path[-] (n5s0) edge (n5s1) (n5s1) edge (n5s2) (n5s2) edge (n5s3) (n5s3) edge (n5s4);
\path[-] (n6s0) edge (n6s1) (n6s1) edge (n6s2) (n6s2) edge (n6s3) (n6s3) edge (n6s4);
\path[-] (n7s0) edge (n7s1) (n7s1) edge (n7s2) (n7s2) edge (n7s3) (n7s3) edge (n7s4);

\path[-] (n0s1) edge (n1s1);
\path[-] (n2s1) edge (n3s1);
\path[-] (n4s1) edge (n5s1);
\path[-] (n6s1) edge (n7s1);

\path[-] (n0s2) edge (n2s2);
\path[-] (n1s2) edge (n3s2);
\path[-] (n4s2) edge (n6s2);
\path[-] (n5s2) edge (n7s2);

\path[-] (n0s3) edge (n4s3);
\path[-] (n1s3) edge (n5s3);
\path[-] (n2s3) edge (n6s3);
\path[-] (n3s3) edge (n7s3);

\node (g_n0s2) [blue dotted, fit = (n0s2) (n1s2) (n2s2) (n3s2) (n0s0)] {};
\node (g_n1s2) [blue dotted, fit = (n4s2) (n5s2) (n6s2) (n7s2) (n7s0)] {};

\node (g_n0s3) [blue dotted, fit = (n0s3) (n1s3) (n2s3) (n3s3) (n4s3) (n5s3) (n6s3) (n7s3)] {};

\node[coordinate] (base) at ($(n3s0)!0.5!(n4s0)$) {};

\end{tikzpicture}}
  \quad
  \subfloat[Fast-SSC Tree]{\label{fig:fast-ssc-tree} \definecolor{deepgreen}{RGB}{8, 130, 25}

\begin{tikzpicture}[baseline=(3_0.center),
        level/.style={level distance = 6mm},
        level 1/.style={sibling distance=19mm, edge from parent/.style={draw,black,line width=2pt}},
        level 2/.style={sibling distance=9mm, edge from parent/.style={draw,black,line width=1pt}},
        level 3/.style={sibling distance=4mm, edge from parent/.style={draw,black,line width=0.5pt}},
        ]

\tikzset{
frozen/.style={thick,draw=black,fill=white,minimum size=3mm,circle, inner sep=0},
fullspace/.style={thick,draw=black,fill=black,minimum size=3mm,circle, inner sep = 0},
mixed/.style={thick,draw=black,fill=gray,minimum size=3mm,circle, inner sep = 0},
rep_mixed/.style={thick,draw=black,pattern=north west lines,pattern color=deepgreen,minimum size=3mm,circle, inner sep = 0},
spc_mixed/.style={thick,draw=black,pattern=crosshatch,pattern color=orange,minimum size=3mm,circle, inner sep = 0},
repspc/.style={thick,draw=black,pattern=vertical lines,pattern color=blue,minimum size=3mm,circle, inner sep = 0}
}

\tikzset{
parallel segment/.style={
   segment distance/.store in=\segDistance,
   segment pos/.store in=\segPos,
   segment length/.store in=\segLength,
   to path={
   ($(\tikztostart)!\segPos!(\tikztotarget)!\segLength/2!(\tikztostart)!\segDistance!90:(\tikztotarget)$) -- 
   ($(\tikztostart)!\segPos!(\tikztotarget)!\segLength/2!(\tikztotarget)!\segDistance!-90:(\tikztostart)$)  \tikztonodes
   }, 
   segment pos=.5,
   segment length=4ex,
   segment distance=-1mm,
},
}

\node[mixed] (3_0){} [grow=left]
	child {node[rep_mixed, label={above:{\footnotesize Repetition}}] (2_0){}
	}
	child {node[spc_mixed, label={below:{\footnotesize SPC}}] (2_1){}
	}
;
	
\draw[->] (3_0) to[parallel segment] node[above right=-1.5mm] {\footnotesize 1} (2_0);
\draw[->] (2_0) to[parallel segment] node[below left=-1.5mm] {\footnotesize 2} (3_0);

\draw[->] (3_0) to[parallel segment] node[above left=-1.5mm] {\footnotesize 3} (2_1);
\draw[->] (2_1) to[parallel segment] node[below right=-1.5mm] {\footnotesize 4} (3_0);

\draw[<-] ($(3_0.east) + (0mm, 0.5mm)$) -- node[above=0mm] {\footnotesize Channel} ($(3_0.east) + (10mm, 0.5mm)$);
\draw[->] ($(3_0.east) + (0mm, -0.5mm)$) -- node[below=0mm] {\footnotesize Codeword} ($(3_0.east) + (10mm, -0.5mm)$);

\end{tikzpicture}}
  \caption{The graph of an (8, 4) polar code and its corresponding Fast-SSC tree representation.}
  \label{fig:fast-ssc}
\end{figure}

\subsection{Stage Indices and Sizes}
\label{sec:hw-dec:stages}
The Fast-SSC decoder in \cite{Sarkis2014} starts decoding a polar code of length $n_{\text{max}}$ at stage $S_{\log n _{\text{max}}}$; where a stage $S_i$ corresponds to a constituent polar code of length $2^i$, as discussed in Section~\ref{sec:bg:fast-ssc}. Since the decoder uses a semi-parallel architecture, the length of the constituent code is used to determine the number of memory words associated with a stage.
The simplest method for that decoder to decode a code of length $n \leq n_{\text{max}}$ is to store the channel LLRs in the memory associated with stage $S_n$ and start the decoding process from there.
This, however, requires significant routing resources since the architecture presented in \cite{Sarkis2014} separates the channel and internal LLRs into different memories for performance reasons.

In the proposed flexible decoder, we calculate the length, $n_v$, of the constituent code associated with a stage $S_i$ as function of $i$, $n$, and $n_{\text{max}}$:
\begin{equation}
n_v(S_i) = 2^i \frac{n}{n_{\text{max}}}.
\label{eq:hw-flex-n}
\end{equation}

The memory allocated for a stage $S_i$ remains the same regardless of $n$ and is always calculated assuming $n = n_{\text{max}}$.
The flexible decoder always starts from $S_{\log n_{\text{max}} }$, corresponding to a polar code of length $n \leq n_{\text{max}}$, performing operations on $n_v(S_i)/(2\mathcal{P})$ LLR values at a stage $S_i$, and proceeds until it encounters a constituent code whose output can be directly estimated according to the rules of the Fast-SSC algorithm.

\subsection{Implementation Results}
\label{sec:hw-dec:impl}
Since memory is accessed as words containing multiple $2\mathcal{P}$ LLR or bit-estimate values, the limits used to determine the number of memory words per stage must be changed to accommodate the new $n$ value. The rest of the decoder implementation remains unchanged from \cite{Sarkis2014}. These limits are now calculated according to \eqref{eq:hw-flex-n} and using the $n$ value provided to the decoder as an input.

Table~\ref{tab:hw:impl-n-max} compares the proposed flexible decoder ($n_{\text{max}} = 32,768$) with the Fast-SSC decoder of \cite{Sarkis2014} ($n = 32,768$) when both are implemented using the Altera Stratix IV EP4SGX530KH40C2 FPGA. The resource requirements are also provided for $n_{\text{max}} = n = 2048$. It can be observed that the change in resource utilization is negligible as a result of the localized change in limit calculations.

When decoding a code of length $n < n_{\text{max}}$, the flexible decoder has the same latency (in clock cycles) as the Fast-SSC decoder for a code of length $n$. Since our $Rn$-flexible decoder implementation has the same operating clock frequency as the $R$-flexible decoder, it also has the same throughput and latency (in time). We note that the decoders presented in this work contain an additional input buffer to store an incoming channel vector while one is being decoded. This is done to enable loading-while-decoding and allows the decoder to sustain its throughput.

The implementation results for $Rn$-flexible decoder supporting code shortening are also included in Table~\ref{tab:hw:impl-n-max}. The main change is the requirement for $n_\text{max}$ more bit of RAM. This is a consequence of shortening being implemented using masking where the LLRs corresponding to shortened bits are replaced with the maximum LLR value based on an $n_\text{max}$-bit mask that is stored in said memory.

\begin{table}[t]
  \centering
  \caption{Implementation of the $Rn$-flexible polar decoder compared to the $R$-flexible decoder of \cite{Sarkis2014} on the Altera Stratix IV EP4SGX530KH40C2.}
  \begin{tabular}{r c c c c c c}
    \toprule
    Decoder & $n$ & $\mathcal{P}$ & LUTs & FF & RAM    & $f$\\
            &     &               &      &    & (bits) & (MHz)\\
    \midrule
    \cite{Sarkis2014} & 2048 & 64 & 6315 & 1608 & 50,072 & 102\\
    Proposed          & 2048 & 64 & 6507 & 1600 & 50,072 & 102\\
    w/ Shortening     & 2048 & 64 & 6451 & 1613 & 52,120 & 102\\
    \midrule
    \cite{Sarkis2014} & 32,768 & 256 & 24,066 & 7,231 & 536,136 & 102\\
    Proposed          & 32,768 & 256 & 23,583 & 7,207 & 536,136 & 102\\
    w/ Shortening     & 32,768 & 256 & 23,593 & 7,219 & 568,904 & 102\\
    \bottomrule
  \end{tabular}
  \label{tab:hw:impl-n-max}
\end{table}

\section{Flexible Software Decoders}
\label{sec:sw-dec}
High-throughput software decoders require vectorization using SIMD instructions in addition to a reduction in the number of branches. However, these two considerations significantly limit the flexibility of the decoder to the point where the lowest latency decoders in literature are compiled for a single polar code \cite{Giard_TSP_2015}. In this section, we present a software Fast-SSC decoder balancing flexibility and decoding latency. The proposed decoder has 37\% higher latency than a fully specialized decoder, but can decode any polar code of length $n \leq n_{\text{max}}$. As will be discussed later in this section, there are two additional advantages to the proposed flexible software decoder: the resulting executable size is an order of magnitude smaller, and it can be used to decode very long polar codes for which an unrolled decoder cannot be compiled.
Since SIMD instructions operate mostly `vertically' on data stored in different vectors, natural indexing is preferable to reversed one; in contrast to the hardware decoders.

\subsection{Memory}
\label{sec:sw-dec:len}
Unlike in hardware decoders, it is simple to access an arbitrary memory location in software decoders. The LLR memory in the proposed software decoder is arranged into stages according to constituent code sizes. When a code of length $n \leq n_{\text{max}}$ is to be decoded, the channel LLRs are loaded into stage $S_{\log n}$, bypassing any stages with a larger index.

When backtracking through the code tree towards stages with high indices, the decoder performs the same operations on bit-estimates as the encoder---namely, binary addition and copying, depending on the index of the output bit in question. Storing the bit-estimates in a one-dimensional array of length $n_{\text{max}}$ bits enables the decoder to only perform the binary addition and store its results, eliminating superfluous copy operations and decreasing latency \cite{Giard_TSP_2015}. For a code of length $n \leq n_{\text{max}}$, the decoder writes its estimates starting from bit index 0. Once decoding is completed, the estimated codeword will occupy the first $n$ bits of the bit-estimate memory, which are provided as the decoder output.

\subsection{Vectorization}
\label{sec:sw-dec:vec}
The unrolled software decoder \cite{Giard_TSP_2015} specifies input sizes for each command at compile time. This enables SIMD vectorization without any loops, but limits the decoder to a specific polar code. To efficiently utilize SIMD instructions while minimizing the number of loops and conditionals, we employ dynamic dispatch in the proposed decoder. Each decoder operation in implemented, using SIMD instructions and C++ templates, for all stage sizes up to $n_{\text{max}}$. These differently sized implementations are stored in array indexed by the logarithm of the stage size. Therefore two branch operations are used: the first to look up the decoding operation, and the second to look up the correct size of that operation. This is significantly more efficient than using loops over the SIMD word size.

\subsection{Results}
\label{sec:sw-dec:results}
We compare the latency of the proposed vectorized flexible decoder with a non-vectorized version and with the fully unrolled decoder of \cite{Giard_TSP_2015} using floating-point values.

Table~\ref{tab:sw:speed-n-max} compares the proposed flexible, vectorized decoder with a flexible, non-explicitly-vectorized decoder (denoted `scalar') and a fully unrolled (denoted `unrolled') one running on an Intel Core i7-2600 with AVX extensions. All decoders were decoding a (32768, 29492) polar code using the Fast-SSC algorithm, floating-point values, and the min-sum approximation. The flexible decoders had $n_{\text{max}} = 32,768$.
From the results in the table, it can be seen that the vectorized decoder has 42.6\% the latency (or 2.3 times the throughput) of the scalar version. Compared to the code-specific unrolled decoder, the proposed decoder has 137\% the latency (or 73\% the throughput).
In addition to the two layers of indirection in the proposed decoder, the lack of inlining contributes to this increase in latency. In the unrolled decoder, the entire decoding flow is known at compile time, allowing the compiler to inline function calls, especially those related to smaller stages. This information is not available to the flexible decoder.

Results for $n < n_{\text{max}}$ are shown in Table~\ref{tab:sw:speed-n} where $n_{\text{max}} = 32,768$ for the flexible decoders and the code used was a (2048, 1723) polar code. The advantage the vectorized decoder has 68\% the latency of the non-vectorized decoder. The gap between the proposed decoder and the unrolled one increases to 2.8 times the latency. These decrease in relative performance of the proposed decoder is a result of using a shorter code where a smaller proportion of stage operations are inlined.

In addition to decoding different codes, the proposed flexible decoder has an advantage over the fully unrolled one in terms of resulting executable size and the maximum length of the polar code to be decoded. The size of the executable corresponding to the proposed decoder with $n_{\text{max}} = 32,768$ was 0.44 MB with 3 KB to store the polar code instructions in an uncompressed textual representation; whereas that of the unrolled decoder was 3 MB. In terms of polar code length, the GNU C++ compiler was unable to compile an unrolled decoder for a code of length $2^{24}$ even with 32 GB of RAM; while the proposed decoder did not exhibit any such issues.

\begin{table}[t]
  \centering
  \caption{Speed of the proposed vectorized decoder compared with that of non-vectorized and fully-unrolled decoders when $n = n_{\text{max}} = 32768$ and $k = 29492$.}
  \begin{tabular}{l c c}
    \toprule
    Decoder & Latency ($\mu$s) & Info. Throughput (Mbps)\\
    \midrule
    Scalar Fast-SSC & 256 & 115\\
    Unrolled Fast-SSC \cite{Giard_TSP_2015} & 109 & 270\\
    Proposed Fast-SSC & 149 & 198\\
    \bottomrule
  \end{tabular}
  \label{tab:sw:speed-n-max}
\end{table}

\begin{table}[t]
  \centering
  \caption{Speed of the proposed vectorized decoder compared with that of non-vectorized and fully-unrolled decoders for a (2048, 1723) code and $n_{\text{max}} = 32768$.}
  \begin{tabular}{l c c}
    \toprule
    Decoder & Latency ($\mu$s) & Info. Throughput (Mbps)\\
    \midrule
    Scalar Fast-SSC & 16.2 & 106\\
    Unrolled Fast-SSC \cite{Giard_TSP_2015} & 4.0 & 430 \\
    Proposed Fast-SSC & 11.1 & 155\\
    \bottomrule
  \end{tabular}
  \label{tab:sw:speed-n}
\end{table}

\section{Conclusion}
In this work, we studied the flexibility in code rate and length of polar encoders and decoders. We proved the correctness of a flexible, parallelizeable systematic polar encoding algorithm and used it to implement high-speed, low-complexity hardware encoders with throughput up to 29 Gbps on FPGA. The proof of correctness was provided not only for polar, but also for Reed-Muller codes. Software encoders were also presented and shown to achieve throughput up to 10 Gbps. We demonstrated rate and length flexible hardware decoders that had similar implementation complexity and the same speed as their rate-only flexible counterparts. Finally, we introduced software decoders that are flexible and able to achieve 73\% the throughput of their unrolled, code-specific counterparts.

\bibliographystyle{IEEEtran}
\bibliography{IEEEabrv,systematic-encoder.bib}

\end{document}